\documentclass[letterpaper,USenglish,cleveref]{lipics-v2021}

\usepackage{macros}

\title{Convergent Gate Elimination and Constructive Circuit Lower Bounds} 

\author{Marco Carmosino}{MIT-IBM Watson AI Lab, Cambridge, MA, USA}{mlc@ibm.com}{https://orcid.org/0009-0007-1118-1352}{}

\author{Ngu Dang}{Boston University, Boston, MA, USA}{ndang@bu.edu}{https://orcid.org/0009-0004-2774-2247}{}

\author{Tim Jackman}{Boston University, Boston, MA, USA}{tjackman@bu.edu}{https://orcid.org/0000-0002-2293-5670}{}

\authorrunning{M. Carmosino, N. Dang, and T. Jackman} 

\Copyright{Marco Carmosino, Ngu Dang, and Tim Jackman} 



\category{} 

\relatedversion{} 
\supplement{\supplementdetails[linktext={Machine-Readable Convergence Proofs},  swhid={Software Heritage Identifier}]{General Classification (e.g. Software, Dataset, Model, ...)}{https://feasible-math.org/} }

\acknowledgements{}

\EventEditors{John Q. Open and Joan R. Access}
\EventNoEds{2}
\EventLongTitle{42nd Conference on Very Important Topics (CVIT 2016)}
\EventShortTitle{CVIT 2016}
\EventAcronym{CVIT}
\EventYear{2016}
\EventDate{December 24--27, 2016}
\EventLocation{Little Whinging, United Kingdom}
\EventLogo{}
\SeriesVolume{42}
\ArticleNo{23}
\begin{document}
\nolinenumbers
\maketitle

\begin{abstract}
  Gate elimination is the only method known that can prove complexity lower bounds for explicit functions against unrestricted Boolean circuits.  After decades of effort, the best results are still linear: Li and Yang proved lower bounds of $3.1n - o(1)$ in the basis $\mathcal{B}_2$ consisting of all binary Boolean functions (STOC 2022), and Iwama and Morizumi proved lower bounds of $5n - o(1)$ in the DeMorgan basis consisting only of AND, OR, and NOT functions (MFCS 2002).  These arguments deploy intricate and lengthy case analyses that substitute carefully-chosen inputs into circuits and track which gates become redundant -- thus, ``gate elimination''.  Unlike other methods in complexity theory, there is only limited characterization of the obstacles to improving circuit lower bounds proved via gate elimination.

  Towards better understanding of gate elimination, this work contributes: (1) formalizing circuit simplifications as a convergent term graph rewriting system and (2) giving a simple and constructive proof of a classical lower bound using this system.
  
  First, we show that circuit simplification is a convergent term graph rewriting system over the DeMorgan and $\{\land, \lor, \oplus\}$ bases.  We define local rewriting rules from Boolean identities such that every simplification sequence yields an identical final result (up to circuit isomorphism or bisimulation).  Convergence enables rigorous reasoning about structural properties of simplified circuits without dependence on the order of simplification.  Then, we show that there is \emph{no similar} convergent formalization of circuit simplification over the $\cU_2$ and $\cB_2$ bases.
  
  Then, we use our simplification system to give a constructive circuit lower bound, generalizing Schnorr's classical result that the XOR function requires $3(n - 1)$ gates to compute in the DeMorgan basis.  A constructive lower bound $f \not\in \mathcal{C}$ gives an algorithm (called a ``refuter'') that efficiently finds counter-examples for every $\mathcal{C}$-circuit trying to compute the function $f$.  Chen, Jin, Santhanam, and  Williams showed that constructivity plays a central role in many longstanding open problems about complexity theory (FOCS 2021), so it is natural to ask for constructive circuit lower bounds from gate elimination arguments.  This demonstrates how using convergent simplification can lead to shorter and more modular proofs of circuit lower bounds. Furthermore, until this work, no constructive lower bound had been proved via gate elimination.  
\end{abstract}

\newpage

\section{Introduction}
\label{sec:intro}

Circuits model the computation of Boolean functions on fixed input lengths by acyclic wires between atomic processing units -- logical ``gates.''  To measure the circuit complexity of a function $f$, we fix a set of gates $\cB$ -- called a \emph{basis} -- and count the number of binary gates required to compute $f$. Some of the well-known bases include the DeMorgan basis (binary AND, binary OR, and unary NOT gates), the $\cU_2$ basis (all binary Boolean functions except parity and equivalence), and the $\cB_2$ basis (all binary Boolean functions).

Shannon proved in 1949 -- using a simple counting argument -- that most $n$-input Boolean functions require near-maximal circuit complexity ($2^n/n$ gates) to compute \cite{Shannon49}.  This is almost the trivial ``memorization''  complexity obtained for any $n$-input Boolean function by writing a DNF using $O(n2^n)$ gates.  Unfortunately, counting arguments only establish that such hard functions exist.  Identifying \emph{specific} Boolean functions with high circuit complexity is a major open problem, inextricably connected to the ongoing search for efficient pseudorandom bit generators via celebrated \emph{hardness-randomness tradeoffs:} strong enough circuit complexity lower bounds imply $\P = \BPP$ \cite{DBLP:journals/jcss/NisanW94,DBLP:conf/stoc/ImpagliazzoW97,DBLP:journals/cc/KabanetsI04}.  The past seven decades have seen impressive super-polynomial complexity lower bounds proved for circuits \emph{restricted to constant depth} \cite{DBLP:journals/mst/FurstSS84, DBLP:conf/stoc/Smolensky87, DBLP:journals/jacm/Williams14}.  However, complexity lower bounds for \emph{unrestricted} circuits are required for efficient universal derandomization, and remain a substantial challenge.

\emph{Gate elimination} is the most successful technique known for proving complexity lower bounds on \emph{specific} Boolean functions against \emph{unrestricted} circuits over the DeMorgan, $\mathcal{U}_2$, \cite{Redkin1973, Schnorr74, Zwick91, LachishR01, IwamaM02} and $\mathcal{B}_2$ bases \cite{Schnorr74, Stockmeyer77, DemenkovK11, FindGHK2016, Li022}.  The idea is straightforward: input variables are substituted (by constants or by other functions) and then the resulting circuit is simplified. During simplification, complexity measures like the number of gates and assertions about circuit structure are tracked and lifted to lower bounds using induction. For example, if restriction yields a desirable subfunction (e.g., the same function on fewer variables) while removing $c$ gates, then we can inductively argue that the original circuit had at least $c \cdot n$ gates. The simplest gate elimination arguments follow this schematic exactly, substituting constants and removing gates according to basic Boolean identities. Schnorr's tight $3(n-1)$ bound for $\XOR_n$ (and $\neg\XOR_n$) \cite{Schnorr74} exemplifies this nicely: substituting for an input with a constant allows removal of at least three gates and results in a circuit computing $\XOR_{n-1}$ (or $\neg \XOR_{n-1}$). Induction shows that any $\XOR_n$ circuit must have at least $3(n-1)$ gates.

However, progress has been slow and the best known lower bounds remain linear: about $5n$ for the DeMorgan and $\mathcal{U}_2$ bases \cite{LachishR01, IwamaM02} and $3.1n$ for the $\mathcal{B}_2$ basis \cite{Li022}.  Why?  Ongoing attempts to improve these results and understand both the capabilities and limitations of gate elimination are an important research direction.  In a exciting recent development, breakthrough ETH-hardness results used gate elimination to characterize and exploit the \emph{structure} of the linear-size circuits \cite{Ilango20}.  Combined with new results showing that constructive separations (efficient algorithms generating inputs where an incorrect algorithm errs, defined below) are \emph{necessary} for widely-conjectured complexity lower bounds \cite{DBLP:journals/theoretics/00010S024}, this raises natural questions.  Are proofs by gate elimination constructive?  If so, is that an obstacle or an asset?

The questions above motivate this work.  Despite the fundamental role of gate elimination, our understanding of this methodology remains informal and all known arguments are ad-hoc.  We present two contributions:
\begin{enumerate}
    \item Formalization of the circuit-simplification process used in all gate elimination arguments as a \emph{convergent term graph rewriting system.}
    
    \item A \emph{constructive separation} of $\XOR_n$ from circuits of size $< 3(n - 1)$, derived by re-writing Schnorr's gate elimination proof using our convergent circuit simplification system.
\end{enumerate}

\subsection{Our Results \& Research Directions}

\subparagraph*{Convergent Systems for Gate Elimination.}

In Section \ref{sec:tgrs-details}, we develop a \emph{formal} foundation for gate elimination as a \textbf{\emph{convergent} term graph rewriting system} using ideas from term rewriting \cite{BaaderN1998} and graph-based computation models \cite{Plump99}. Sections \ref{sec:boolean-identities} and \ref{sec:term-rw} define simplification systems for Boolean formulas over the DeMorgan and $\{\land, \neg, \oplus\}$ bases where rules are derived from carefully-chosen Boolean identities.  Then, the Knuth-Bendix algorithm \cite{KNUTH1970, HUET1981} is applied to produce convergent formula simplification system. In Section \ref{sec:term-graph-rw}, these \emph{formula} rewriting systems are lifted into \emph{circuit} rewriting systems via Plump's term graph rewriting framework \cite{Plump99}. The resulting systems are \emph{convergent}, that is, any valid sequence of simplifications terminates in a unique normal form, regardless of order.

\begin{theorem*}[Informal Statement of Theorem \ref{thm:GE-confluent}]
    Circuit simplification is convergent in the DeMorgan ($\{\land, \lor, \neg\}$) and $\{\land, \oplus, \neg\}$ bases.
\end{theorem*}

This convergence property is useful for making gate-elimination arguments both clean and automatable: because every simplification sequence terminates in the same normal form, the structural properties of the simplified circuit are \emph{invariant} under the choice of rewrites, letting us reason about circuits ``up to normal form'' rather than requiring justification for each possible simplification path. This is especially important for computer-aided proofs -- non-terminating rewrite systems may not even halt, whereas convergent rewriting guarantees completion and makes brute-force exploration unnecessary.  We hope that casting gate elimination in a convergent term-graph rewriting framework contributes to ongoing formalization of complexity theory (see Section \ref{sec:related-work} on Related Work).

\begin{direction} 
    Formalize known gate elimination arguments in standard proof assistants with the purpose of developing new computer-assisted gate elimination arguments and ultimately improved circuit lower bounds.
\end{direction}

To further this goal, we reproved Schnorr’s lower bound for $\XOR_n$ \cite{Schnorr74} using the system (Appendix \ref{sec:ckt-rewrite-demo}) and in a structured format, aiming to ease formalization (Appendix \ref{sec:schnorr-structured-proof}). Verifying convergence of our system and Schnorr's argument using a proof assistant like Lean4 is a natural next step.

\subparagraph*{Non-confluent Gate Elimination.}
We explore varying our systems in Section \ref{sec:u2-ge-not-confleunt} to surprising effect. Classical circuit complexity arguments conflate the DeMorgan basis and $\cU_2$ as interchangeable. This treatment is not baseless: with respect to size and depth bounds, the two bases are identical for all non-degenerate functions except $f(x) = \neg x$. This equivalence also holds between $\{\neg, \land, \oplus\}$ and $\cB_2$; arguments like \cite{Paul75} begin by categorizing $\cB_2$ gates as ``AND''-like and ``XOR''-like, informally changing the basis to $\{\neg, \land, \oplus\}$. We make these folklore connections between the bases explicit in Appendix \ref{sec:demorgan-u2-equivalence}. However, we find that when formalizing gate elimination as a rewriting system, this equivalence between the bases immediately breaks down.

\begin{observation*}[Informal Statement of Observation \ref{obs:u2-b2-not-confluent}]
    $\cU_2$ and $\cB_2$ do not admit confluent circuit simplification systems.
\end{observation*}

This failure arises from the presence of superfluous negation gates which need to be subsumed upwards and/or downwards during simplification, causing divergent simplification paths and ultimately non-isomorphic circuits (see Figure \ref{fig:pushing-example} in Section \ref{sec:u2-ge-not-confleunt}). 

Even different bases are often treated similarly since translations often only sees constant changes to size and depth. As researchers are interested in proving non-linear lower bounds, this cost is seemingly inconsequential. This results highlights how impactful a basis choice can be; arguments in the DeMorgan basis or $\{\land, \oplus, \neg\}$ can be \emph{easier} than their $\cU_2$ or $\cB_2$ counterparts. This is especially true if implemented in a proof assistant -- formalizing in these bases may cut down the search space depending on the specific argument. 

By casting gate elimination as a rewriting system, we gain a better understanding of its underlying properties. Besides decades of limited progress and known limitations (see Related Works), provable barriers to using gate elimination to prove superlinear circuit bounds remain unknown.

\begin{direction}
    Further investigate gate elimination as a rewriting system. How does varying bases and rules affect its power? Does there exist a barrier for the technique against proving circuit lower bounds?
\end{direction}

\subparagraph*{Constructive Circuit Lower Bounds.}
The statement ``Boolean function $f$ does not have size $s(n)$ circuits'' has quantifier structure $\forall\exists$ when translated into logic:
\[
\forall n \in \nat ~~ \forall C \text{ such that } \mathsf{size}(C) < s(n) ~~ \exists x \in \bool^n ~~ C(x) \neq f(x).
\]
Concretely, every circuit that fails to compute $f$ must err on a particular input $x$.  Thus, every circuit complexity lower bound for a specific function $f$ induces a natural total search problem: given as input a circuit $C$ that is too small (and therefore \emph{must} fail to compute $f$) print a bitstring $x$ such that $C$ and $f$ disagree.  Exciting recent work discovered that constructive separations play a surprisingly central role in complexity theory, by showing that most conjectured uniform separations (including $\P \neq \NP$, $\ZPP \neq \EXP$, and $\BPP \neq \NEXP$) can be automatically transformed into constructive separations \cite{DBLP:journals/theoretics/00010S024}.

On the other hand, they also showed that if certain \emph{well-known} lower bounds (including ``deciding palindromes requires super-quadratic time on one-Tape Turing Machines'', \cite{DBLP:conf/stoc/Maass84}) can be made constructive, then major complexity breakthroughs (like $\P \neq \NP$) would follow \cite{DBLP:journals/theoretics/00010S024}.  Essentially, the intuition that ``simple'' proofs of complexity lower bounds should be ``easy'' to make constructive is wildly inaccurate.  We show that, at least in the case of gate elimination, simple proofs \emph{can} be converted into constructive separations.

\begin{theorem*}[Informal Statement of Theorem \ref{thm:xor-refuter}]
Let $C$ be any DeMorgan circuit on $n$ inputs with fewer than $3(n-1)$ binary gates.  There is an efficient and deterministic algorithm that, given input $C$, prints an $n$-bit string $x$ such that $C(x) \neq \XOR_n(x)$.
\end{theorem*}

This immediately raises further questions: what about other circuit lower bounds proved using gate elimination?  Though some employ sophosticated combinatorial objects, they still use a simple inductive simplification structure.  What about other notions of ``constructivity''?  The efficient error generation algorithms discussed above are not \emph{logical} notions becuase the proof itself is not treated as a mathematical object.  Subsequent work showed that many features of constructive separations in the \emph{algorithmic} sense still hold when we consider proofs in \emph{weak fragments of Peano Arithmetic} closely associated with computational complexity classes, like Cook's $\lang{PV}$ (which is connected to deterministic polynomial time) \cite{DBLP:conf/stoc/GrosserC25}.  There is a recent turn towards understanding the proof-theoretic strength required to establish known theorems and resolve open problems of complexity theory, and it is entirely unknown where gate elimination fits into this landscape \cite{DBLP:journals/sigact/Oliveira25,DBLP:conf/focs/0001LO24}.  The foregoing discussion motivates

\begin{direction}
   Analyze gate elimination using both ``algorithmic'' and logical proof theory.  Which gate elimination proofs can be transformed into constructive separations?  Can weak logical theories (like $\lang{PV}$) prove circuit lower bound via gate elimination?
\end{direction}

\subsection{Related Work} \label{sec:related-work}

\subparagraph*{Proving \& Exploiting Structural Characterization of Circuits.}
Beyond the problem of proving tight circuit bounds for a particular Boolean function lies a natural \emph{design} question  which asks ``what is the structure of its optimal circuit?'' Though this question empirically appears difficult to answer in general \cite{Wegener1987}, it has been investigated for a handful of functions and bases over the last few decades. The primary method here is also gate elimination: one proves that any deviation from the desired structure would allow one to eliminate ``too many'' gates, contradicting optimality. Some of the earliest work on this question include \cite{Sattler81} and \cite{BlumS84} which investigated when optimal circuits for $2$-output Boolean functions compute each output independently of the other. Subsequent work characterized optimal circuits computing $\XOR$ in the DeMorgan basis under the convention that NOT-gates \emph{contribute} to circuit size \cite{Kombarov2011}, and later extended these ideas to other bases \cite{Kombarov2018}. More recently, Rahul Ilango \cite{Ilango20} characterizes the optimal structure of circuits computing $\bigvee_{i=1}^n(x_i\land y_i)$, and leverages this to obtain a breakthrough hardness reduction: the partial function minimum circuit size problem ($\MCSP^*$) is hard assuming the exponential time hypothesis (ETH).

Structural characterizations also have practical implications, but their recent applications highlight a broader theoretical payoff. In particular, Carmosino, Dang, and Jackman \cite{CarmosinoDJ2025} explicitly target and characterize optimal circuits of functions which enjoy linear sized circuits such as the exclusive-or function $\XOR$ where they show that the result of \cite{Kombarov2011} also holds when NOT-gates do \emph{not} contribute to circuit size under the same basis. They further show that the structure of optimal circuits for \emph{simple extensions} (function extensions with minimal additional circuit complexity) enables a fixed-parameter tractable algorithm to identify them. Together, these developments motivate studying gate elimination even if the technique seems unable to prove superlinear bounds: understanding its intrinsic strengths and limitations as a method for proving and exploiting structural theorems may be valuable independently of its role as a lower-bound tool.

These new applications of gate elimination motivate its study even if the technique seems unable to prove superlinear bounds. Thus, we believe that further study is required to understand the strengths and weaknesses inherent to gate elimination as a technique itself, rather than as a tool for proving lower bounds. 

\subparagraph*{Formalizing Circuit Complexity with Proof Assistants.}
Besides its theoretical importance, circuit complexity has clear practical relevance; circuit minimization, synthesis, and verification are central tasks in hardware design. These applications have motivated the creation of a rich suite of software for the design of Boolean circuits \cite{BraytonM10,SoekenRHM18,ReichlSS23,TangZZCLYX25}. Recent open-source tools like Cirbo \cite{AverkovBEGKKKLL25} employ SAT-based techniques and optimal circuit databases to enable efficient optimization, analysis and experimentation \cite{GoncharovKL2025}. While these tools are useful for upper bounds, they do not provide methods for lower bounds beyond exhaustive search. 

The use of theorem provers in studying circuit complexity has mostly been limited to verification. \cite{ShirazH18} provides a HOL4 library which verifies common functional elements and can structurally analyze complex circuits which are composed of these widgets. Rocq has also been used for the verification of classical circuits and quantum circuits \cite{Braibant11, RandPZ17}. One of the few formally verified lower bounds -- Clique cannot be computed by polynomial size monotone circuits -- was formalized in Isabelle/HOL \cite{Thiemann2022}. However, to the best of our knowledge, this work is the first to model circuits as term graphs. This representation is not impracticable for proof assistants -- \cite{WebbHU23} recently formally verified term graph optimizations in Isabelle/HOL. By formally studying gate elimination, we hope to build up the set of tools for verifiable circuit lower bounds and bridge the gap with the extensive upper bound tooling.

\subparagraph*{Barriers Against Gate Elimination.} The relativization, natural proofs, and algebrization barriers to improving complexity separations are not known to apply to gate elimination \cite{DBLP:conf/stacs/RenS22}.  Investigating its limitations, \cite{GolovnevHKK18} constructed functions whose circuits are ``resistant'' to gate elimination: any single substitution reduces their complexity by merely a constant.  As gate elimination arguments rely on a small constant number of substitutions, current techniques cannot prove superlinear circuit lower bounds for these functions.  This analysis is limited in two ways: (i) this requires the class to contain their constructions, which is not always given, and (ii), this limit does not say much about using gate elimination for applications besides circuit lower bounds.


\section{Representing Boolean Circuits as Term Graphs}
\label{sec:prelims}

We study general circuits over two bases. The first is the \emph{DeMorgan basis} $\cB_D = \{ \land, \lor, \neg , 0, 1 \}$ of Boolean functions: binary $\land$ and $\lor$, unary $\neg$, and zero-ary (constants) $1$ and $0$.  The second basis is $\{\land, \oplus, \neg, 0, 1\}$, where $\lor$ has been replaced by the binary exclusive-or function $\oplus$. Circuits take zero-ary variables in $X = \{x_1, x_2, \dots , x_n\}$ for some fixed $n$ as \emph{inputs,} and have exactly one distinguished \emph{output} gate.  Usually, circuits are described as directed acyclic graphs (DAGs) with nodes labeled by function symbols or variables and edges representing ``wires'' between the gates.  For every such circuit $C$, acyclicity of the underlying graph ensures that $C$ computes a well-defined $n$-ary Boolean function by evaluating each gate in topological order according to the function it is labeled with (starting with a substitution of concrete bits into the variables) until the output gate has a value. 

The complexity of a Boolean circuit can be measured in a number of ways including size and depth. In this work, \emph{\textbf{circuit size}} is defined as the \emph{number of \textbf{binary} gates} (i.e., $\neg$ gates are ``free''). Under this measure, the circuit complexity of non-degenerate Boolean functions -- the size of the smallest circuits computing them -- are the same in the DeMorgan basis and $\cU_2$, the basis consisting of the set of all binary Boolean operators besides $\oplus$ and its complement $\odot$. The same relationship holds between $\{\land, \oplus, \neg, 0, 1\}$ and $\cB_2$ = $\cU_2 \cup \{\oplus, \odot\}$. We make this relationship explicit in Appendix \ref{sec:demorgan-u2-equivalence}.

To apply results from term graph rewriting, however, we must describe circuits as \emph{hypergraphs}, tuples $C = \langle V_C, E_C, \lab_C, \att_C \rangle$ where $V_C$ and $E_C$ are finite sets of \emph{vertices} (or nodes) and \emph{hyperedges} respectively, $\lab_C : E_C \to \cB \cup X$ is an edge-label function recording the type of each edge, and $\att_C : E_C \to V^{\leq3}_C$ is an attachment function which assigns a non-empty string of vertices $\att_C(e) = v_0 , \dots , v_\ell$ to each hyperedge $e$ such that $|\att_C(e)| = 1 + \operatorname{arity}(\lab_C(e))$.  We call $v_0$ the \emph{result node} of $e$ and $v_1, \dots , v_\ell$ the \emph{argument nodes}, and denote them $\res(e)$ and $\arg(e)$ respectively.

The hypergraph $C$ is a \emph{term graph} if (1) there is a node $\mathrm{root}_G$ from which each node is reachable, (2) $G$ is acyclic -- there is no path where some node occurs twice, and (3) each node is the result node of a unique edge.  In this setting, hyperedges represent logic gates and vertices represent ``wires'' between gates --- essentially dual to the standard encoding of circuits as DAGs.


\section{Gate Elimination as a Convergent Term Graph Rewriting System}
\label{sec:tgrs-details}

An \emph{abstract rewriting system} is just a set of objects $A$ together with a binary relation $\rightarrow$ on $A$ called the \emph{rewrite} relation.  In this section, we will develop a system $\cS$ where $A$ consists of Boolean circuits over the DeMorgan basis and $C \rightarrow C'$ holds when $C$ simplifies to $C'$ via a single step of gate elimination.  To this end, first recall standard notions concerning abstract rewriting systems (see Definition 2.1.3 of \cite{BaaderN1998}).  For elements $a,a' \in A$, write $a \overset{*}{\rightarrow} a'$ to mean that there is a finite path of rewrite steps from $a$ to $a'$.

\begin{definition} The rewrite relation $\rightarrow$ is called
  \begin{description}
  \item[terminating] iff there is no infinite path $a_0 \rightarrow a_1 \rightarrow \dots$ ,
  \item[confluent] iff for every triple $a,b,b' \in A$, if $a \overset{*}{\rightarrow} b$ and $a \overset{*}{\rightarrow} b'$, then there exists $c \in A$ such that both $b \overset{*}{\rightarrow} c$ and $b' \overset{*}{\rightarrow} c$ , and
  \item[convergent] iff it is both confluent and terminating.
  \end{description}
\end{definition}

An object $a \in A$ is \emph{in normal form} if there is no $b$ such that $a \rightarrow b$.  Objects in normal form often enjoy nice structural properties, because they are ``simplest'' according to the rewrite relation.  For example, circuits $C$ in normal form according to our system $\cS$ will have:

\begin{itemize}
  \item No double negations.
  \item No constants unless $C$ computes a constant function --- e.g., $(x_i \lor \neg x_i)$ normalizes to 1.
  \item No identity gates --- sub-circuits $(\gamma \land \gamma)$ and $(\gamma \lor \gamma)$ do not occur.
\end{itemize}

Each rewrite rule of $\cS$ will (by construction) decrease or preserve the number of gates.  Consequently, it will be straightforward to show that (1) system $\cS$ is terminating and (2) given a circuit $C$ with bits $b = b_1 ,\dots, b_n$ substituted for all input variables, rewriting $C$ using $\cS$ until termination just evaluates $C$ on input $b$. These basic properties connect paths through $\cS$ to efficient circuit manipulation and evaluation. The work in showing that $\cS$ is ``well behaved'' goes into proving confluence.  We proceed along the following lines.

\begin{enumerate}
\item Identify a list of Boolean identities, $\cE_B$, which are sufficient for gate elimination arguments.
\item Use the Knuth-Bendix algorithm on $\cE_B$ to produce a convergent \emph{formula} simplification system $\cR_B$.
\item Lift $\cR_B$ to a convergent \emph{circuit} simplification system $\cS$ via Plump's account of term graph rewriting.
\end{enumerate}

This detailed treatment is for the DeMorgan basis ($\{\land, \lor, \neg\}$); the convergent system for $\{\land, \oplus, \neg\}$ follows similarly albeit with a different set of Boolean identities. 

\subsection{Boolean Identities}\label{sec:boolean-identities}

The identities collected in $\cE_B$ (Definition \ref{def:gate-elim-identities}) describe each step of circuit simplification in proofs by gate elimination.  Clearly, $\cE_B$ is valid for Boolean algebra: for all $g \in \{0, 1\}$, each identity is true when $\approx$ is interpreted as equality on the Boolean domain.  Therefore, consequences derived from $\cE_B$ via ``sound inference rules'' are true.  We do not treat that \emph{equational logic}\footnote{See Chapter 3 of \cite{BaaderN1998} or the exposition of Birkhoff's Theorem in \cite{Plaisted93}.} formally, because we transform $\cE_B$ into a convergent term rewriting system in the next subsection.

\begin{definition}[Useful Boolean Identities]
  \label{def:gate-elim-identities}
  We denote by $\cE_B$ the following set of identities:
  \begin{align*}
    \label{eq:gate-elim-eqn}
    1 \land 1 &\approx 1 & 1 \lor 1 &\approx 1 & \neg 1 &\approx 0   & g \land 1 &\approx g  & g \land 0 &\approx 0 & g \land \neg g &\approx 0  & g \land g &\approx g \\
    1 \land 0 &\approx 0 & 1 \lor 0 &\approx 1 & \neg 0 &\approx 1   & 1 \land g &\approx g  & 0 \land g &\approx 0 & \neg g \land g &\approx 0  & g \lor g &\approx g\\
    0 \land 1 &\approx 0 & 0 \lor 1 &\approx 1 & \neg \neg g &\approx g & g \lor 0 &\approx g  & g \lor 1 &\approx 1 & g \lor \neg g &\approx 1  &   & \\
    0 \land 0 &\approx 0 & 0 \lor 0 &\approx 0 &    &        & 0 \lor g &\approx g  & 1 \lor g &\approx 1 & \neg g \lor g &\approx 1 &   & \\
    \text{(tt }&\text{and)} & \text{(tt }&\text{or)} &\text{(tt }&\text{not)} & \text{(pass}&\text{ing)} & \text{(fix}&\text{ing)} & \text{(taut}&\text{ology)} & \text{(sim}&\text{plify)}
  \end{align*}
\end{definition}

Some Boolean identities are omitted from $\cE_B$, such as commutativity of $\land$ and $\lor$.  This is for two reasons: (1) including them would induce a rewriting system that is \textit{not convergent} and (2) commutativity identities do not ``simplify'' Boolean expressions -- the right hand sides do not have fewer Boolean operators than the left.  While $\cE_B$ does not fully characterize Boolean algebra, it is sufficient to capture the steps of gate elimination.  The identities $\cE_B$ are available in machine-readable form as supplementary material:\\ \url{https://feasible-math.org/2026-FSCD-80-supplementary.zip}

\subsection{Convergent Term Rewriting for Boolean Formulas}
\label{sec:term-rw}

Now we mechanically transform $\cE_B$  into an abstract term rewriting system on Boolean formulas and prove convergence.  Term rewriting is a classical special case of abstract rewriting systems, motivated by algebra, logic, and programming language theory.  \emph{Terms} are treelike expressions built up from function symbols, constants, and variables.  We will take an intermediate step through treelike expressions to get to DAG-like expressions: circuits.  Terms in general and DeMorgan formulas in particular are defined below, along with some auxiliary notions required to specify appropriate rewrite relations.

\begin{definition}[DeMorgan Formulas as Terms]
Let $\Sigma$ be a finite tuple of function and constant symbols with arities $\vec{d} \in \mathbb{N}^{|\Sigma|}$, and let $Z$ denote an infinite set of variables.  $\cT(\Sigma,Z)$ denotes the set of all terms over $\Sigma$ and $Z$, defined inductively:
  \begin{itemize}
  \item Every variable $z \in Z$ is a term.
  \item Every application of a function symbol $f_i \in \Sigma$ to $d_i$ terms $t_1, \dots , t_{d_i}$ of the form $f(t_1,\dots , t_{d_i})$ is a term.
  \end{itemize}
  Let $B = \langle \land , \lor , \neg , 0 , 1 \rangle$ with arities $\langle 2,2,1,0,0 \rangle$ be the \emph{DeMorgan signature}.  We denote by $F = \cT(B,X)$ where $X = \{g, x_1, x_2, \dots \}$ the set of \emph{DeMorgan formulas.}
\end{definition}

A \emph{substitution} $\sigma$ is a mapping between terms that may replace any finite number of variables with another term, but must leave constants and function applications fixed. So, can write substitutions as $\sigma = \{ x_i \mapsto t \}$. A \emph{term rewrite rule} is a pair of terms $\langle \ell, r \rangle$ written as $\ell \rightarrow r$, such that (1) $\ell$ is not a variable and (2) the set of variables in $r$ is a subset of the variables in $\ell$.  A \emph{term rewriting system} (TRS) over $\cT(\Sigma, Z)$ is a set $\cR$ of term rewrite rules where all pairs of terms are from $\cT(\Sigma, Z)$.  Finally, we have:

\begin{definition}[Term Rewriting, Definition 4.1 of \cite{Plump99}]
  The rewrite relation $\rightarrow_\cR$ on $\cT(\Sigma,Z)$ induced by a term rewriting system $\cR$ is defined as follows:\\
  \hspace{1cm} $t \rightarrow_{\cR} u$ if there is a rule $\ell \rightarrow r$ in $\cR$ and a substitution $\sigma$ such that
 \begin{enumerate}
 \item The left-hand side of the rule ``matches'' $t$ --- $\sigma(\ell)$ is a subterm of $t$
 \item The right-hand side ``generates'' $u$ --- $u$ is obtained from $t$ by replacing an occurence of $\sigma(\ell)$ by $\sigma(r)$
 \end{enumerate}
\end{definition}

We can now give the precise type of $\cE_B$: it is a set of pairs of terms from $F$. We aim to transform $\cE_B$ into a convergent term rewriting system $\cR_B$. Rather than manually casting our equations as term rewrite rules (e.g. $g \land 1 \approx g \Longrightarrow g \land 1 \rightarrow g)$ and then proving convergence from scratch, we use a well-known procedure designed to automate this: the Knuth-Bendix completion algorithm \cite{KNUTH1970,HUET1981}.

\begin{theorem}[Knuth-Bendix, \cite{SternagelZ12}]
  \label{thm:Knuth-Bendix}
  Given as input a set of identities $\cE$ over $\cT(\Sigma,Z)$, if Knuth-Bendix terminates, it outputs a convergent term rewriting system $\cR$ over $\cT(\Sigma,Z)$ with the same consequences as $\cE$.
\end{theorem}

At a high level, Knuth-Bendix completion works by ensuring that overlapping pairs of rules do not create ambiguities.  Furthermore, the algorithm may add new rules and/or simplify rules to create a convergent system. Manually running the algorithm in our case would require checking ${\binom{25}{2}}$ pairs of equations --- although not every possible pair overlaps. While this is technically doable by hand, we instead use the open-source Knuth-Bendix Completion Visualizer software (KBCV,  \cite{SternagelZ12}) to run the algorithm on input $\cE_B$, yielding

\begin{definition}[Term Rewriting System $\cR_B$]
\label{def:GE-TRS}
  \begin{align*}
    0     &\to \neg 1 & g \land \neg1  &\to \neg 1 & g \land 1   &\to g  & g \land \neg g &\to \neg1  \\
    \neg\neg g  &\to g   & \neg 1 \land g &\to \neg 1 & 1 \land g   &\to g  & \neg g \land g &\to \neg1  \\
    g \land g &\to g   & g \lor 1   &\to 1   & g \lor \neg 1 &\to g  & g \lor \neg g &\to 1  \\
    g \lor g &\to g   & 1 \lor g   &\to 1   & \neg 1 \lor g &\to g  & \neg g \lor g &\to 1 \\
    \text{(norma}&\text{lizing)} & \text{(fix}&\text{ing)} & \text{(pass}&\text{ing)} & \text{(taut}&\text{ology)}
  \end{align*}
\end{definition}

\begin{lemma}
  \label{lem:simplify-DeMorgan-converge}
  $\cR_B$ is a convergent term rewriting system for the DeMorgan formulas, $F$.
\end{lemma}

\begin{proof}
The Knuth-Bendix algorithm run on input $\cE_B$ terminated and printed the system $\cR_B$ listed in Definition \ref{def:GE-TRS} above.  Therefore, $\cR_B$ is convergent and has the same consequences as $\cE_B$. A machine-verifiable proof has been included in the supplementary material in certification problem format (CPF) which can be verified using tools like \textsf{C\kern-0.15exe\kern-0.45exT\kern-0.45exA} \cite{SternagelT13}. KBCV outputs its proofs in CPF 2 while \textsf{C\kern-0.15exe\kern-0.45exT\kern-0.45exA} requires CPF 3. We have converted the output to CPF 3 using the tool at \url{http://cl-informatik.uibk.ac.at/software/ceta/cpf3.tgz}.
\end{proof}

The rules of $\cR_B$ are grouped based on their structure and impact on a formula.  Observe that $\cR_B$ is a smaller set than $\cE_B$, as Knuth-Bendix has made a few simplifications including removal of truth-table identities.  The unique added rule, $0 \rightarrow \neg 1$, stands out. This is the only rewrite rule in the system $\cR_B$ that increases the number of Boolean operators in a formula.  This is compatible with gate elimination because: (1) the formulas become simpler in the sense that after rewriting $0$ to $\neg 1$  only a single type of constant occurs, and (2) substituting a variable for $0$ can be accomplished by a single additional step of rewriting to $\neg 1$.  The term rewriting rules $\cR_B$ are available in machine-readable form as part of supplementary material.

All that remains is lifting this rewriting system for \textit{formulas} to one for \textit{circuits.}

\subsection{Convergent Term Graph Rewriting for Boolean Circuits}
\label{sec:term-graph-rw}

Recalling Plump's definitions, we can lift our term rewriting system to a term \emph{graph} rewriting system by generalizing the notions of pattern matching and generation as follows \cite{Plump99}.

\begin{definition}[Hypergraph Morphism]
  \label{def:morphism}
  For hypergraphs $G$ and $H$, a \emph{hypergraph morphism} $f: G \to H$ is a pair of functions $f_E : E_G \to E_H$ and $f_V : V_G \to V_H$ that preserve
   \begin{description}
   \item[labels,] so $f_E$ maps each edge $\gamma$ of $G$ to an edge of $H$ with matching label:\\ $ \forall \gamma \in E_G ~\lab_G(\gamma) = \lab_H(f_E(\gamma))$ and
   \item[attachments,] so for every edge $\gamma$ of $G$, $f_V$ is an order-preserving map from the vertices attached to $\gamma$ to the vertices attached to $f_E(\gamma)$ in $H$ --- recalling that $\att(\gamma)$ is a string:\\
    \( \forall \gamma \in E_G ~\forall i \in |\att_G(\gamma) |
     ~~f_{V}(\att_G(\gamma))[i] = \att_H(f_E(\gamma))[i] . \)
   \end{description}
 \end{definition}

 Term graphs can be transformed into terms in the natural way: by ``unrolling'' term graphs so that each edge is read only once.

 \begin{definition}[Term representation of a term graph (3.2 in \cite{Plump99})]
   A node $v$ in a term graph $G$ represents the term
   \(
     \mathrm{term}_G(v) = \lab_G(e)(\mathrm{term}_G(v_1) , \dots , \mathrm{term}_G(v_\ell))
   \), where $e$ is the unique edge with $\res(e) = v$ and where $\arg(v) = v_1 , \dots v_\ell$.  Denote $\mathrm{term}_G(\mathrm{root}_G)$ by $\mathrm{term}(G)$.
 \end{definition}

 Similarly, terms $t$ can be transformed into term graphs by encoding each application of a function symbol in $t$ by a hyper-edge; the resulting hypergraph is essentially the parse tree of $t$.  To formulate pattern matching, Plump extended this simple idea to erase specific variables: for a term $t$, define $\diamond t$ as the parse tree of $t$ encoded by a hypergraph, with all repeated variables collapsed into ``open'' vertices --- that is, the edge $x_i$ is deleted for each $x_i$, but the unique result vertex remains and is referenced by every edge that was attached to $x_i$ in the parse tree.  A term graph $L$ is an \emph{instance} of a term $l$ if there is a graph morphism $\diamond l \to L$ that sends the root of $\diamond l$ to the root of $L$. Given a node $v$ in a term graph $G$ and a term rewrite rule $r \rightarrow \ell$, the pair $\langle v, \ell \to r\rangle$ is a \emph{redux} if the subgraph of $G$ reachable from $v$ (denoted $G[v]$) is an instance of $\ell$.  Finally, we recall a single step of Plump's graph rewriting: essentially, a subgraph matching the left hand side of a rule is sliced out and replaced with the right-hand side.

\begin{definition}[Term Graph Rewriting (Definition 1.4.5 of \cite{Plump99})]
  \label{def:tgrs}
  Let $G$ be a term graph containing a redux $\langle v, l \rightarrow r \rangle$.  There is a  \emph{proper rewrite step} from $G$ to $H$ where $H$ is constructed by
  \begin{enumerate}
  \item $G_1 \gets G - \{e\}$ where $e$ is the unique edge that satisfies $\mathrm{res}(e) = v$
  \item $G_2 \gets$ the disjoint union of $G_1$ with $\diamond r$ where
    \begin{itemize}
    \item $v$ is identified with $\mathrm{root}(\diamond r)$
    \item Every edge labeled with a variable in $\diamond r$ is identified according to the morphism that matched $\ell$ to $G$.
    \end{itemize}
  \item Garbage collection: $H$ is obtained from $G_2$ by deleting all nodes and edges not reachable from the root.
  \end{enumerate}
\end{definition}

We will consider rewriting up to two different notions of equivalence for hypergraphs.  One is hypergraph isomorphism, the other is

\begin{definition}[Bisimilarity]
Two term graphs $G$ and $H$ are \emph{bisimilar}, denoted $G \sim H$, if $\mathrm{term}(G) = \mathrm{term}(H)$.
\end{definition}

Crucially for our applications, two bisimilar circuits compute the same function. Thus, though weaker than isomorphism, bisimilarity is still valuable in the circuit complexity setting: it implies \emph{functional} equivalence.  Plump showed that under both notions of equivalence, convergent term rewriting can be used to manufacture convergent term \emph{graph} rewriting.

\begin{theorem}[Corollary 1.7.4 and Theorem 1.7.12 of \cite{Plump99}]
  \label{thm:ckt-simp-systems}
  If $\cR$ is a convergent term rewriting system, then 
  \begin{enumerate}
  \item  the term graph rewriting system consisting of proper rewrite steps induced by $\cR$ plus the \emph{\textbf{collapse rule}} is convergent up to hypergraph isomorphism, and
  \item the term graph rewriting system consisting only of proper rewrite steps \emph{\textbf{(plain rewriting)}} induced by $\cR$ is convergent up to bisimilairty classes of hypergraphs.
  \end{enumerate}
\end{theorem}

\emph{Collapse} is an additional term graph rewriting rule that allows us to merge two rooted subhypergraphs if there exists a root-preserving hypergraph morphism between them.  For circuits, this operation would correspond to finding redundant subcircuits and combining them into one.  This is natural simplification step; indeed, if our goal was to optimize circuits then any system missing this rule would be insufficient.  In gate elimination proofs the collapse rule is largely irrelevant -- arguments start with sub-optimal circuits and simplify them until a contradiction is obtained to prove circuit complexity lower bounds.

However, it is still desirable to work with collapse because it provides a strong guarantee: circuit simplifications act identically on \emph{isomorphism classes} of circuits.  When aiming for structural characterizations which should hold up to isomorphism, this is the natural choice.  Unfortunately, it is not clear if the collapse rule can be implemented efficiently.  Therefore, we use the plain rewriting system (without collapse rule) when constructive arguments are required (e.g., Theorem  \ref{thm:xor-refuter}).

Applying this lifting theorem to our term rewriting system $\cR_B$ for simplification of DeMorgan formulas yields systems for gate elimination.

\begin{theorem}
\label{thm:GE-confluent}
$\cR_B$ induces two term graph rewriting systems for DeMorgan \emph{circuits:}
\begin{enumerate}
\item system $\cS$ includes the collapse rule and is convergent up to circuit isomorphism, and
\item system $\cS'$ does not include the collapse rule and is convergent up to circuit bisimulation.
\end{enumerate}
\end{theorem}

\subsection{Explicit Circuit Simplification Systems}
\label{sec:ckt-simplification-system-compiled}


In this section we ``compile'' the systems $\cS$ and $\cS'$ to give a concrete description of DeMorgan circuit simplification, complementing the entirely abstract statement of Theorem \ref{thm:GE-confluent}.

There are three parts to $\cS$: (1) notions of redundancy and pattern matching for sub-circuits (Definitions \ref{def:ckt-collapse} and \ref{def:redux}),  (2) a set of circuit rewrite rules given as pairs of patterns (Figures \ref{fig:normalizing-rules} -- \ref{fig:tautology-rules} in Appendix \ref{sec:gate-elim-rules}), and (3) the procedure for pattern-substitution in a circuit (Algorithm \ref{algo:ckt-simp-system}).  System $\cS$ is then the binary relation on circuits induced by setting $C \rightarrow C'$ when either (a) $C$ matches the left-hand side $l$ of a rule $\langle l \mapsto r \rangle$ and $C'$ is the result of substituting pattern $r$ for $l$ in $C$, or (b) $C$ has a redundant sub-circuit that ``collapses'' in $C'$.  System $\cS'$ is identical except for omission of the collapse rule.
 
 \begin{definition}[Collapsing Redundant Sub-Circuits]
   \label{def:ckt-collapse}
   Circuit $C$ \emph{collapses} to circuit $D$ if there is a non-injective hypergraph morphism $C \mapsto D$ mapping $\textrm{root}_C$ to $\textrm{root}_D$, denoted $C \succ D$.
 \end{definition}
 
 \begin{definition}[Pattern \& Redux in Circuits]
   \label{def:redux}
   A \emph{pattern} is just a circuit $L$ where vertices $\gamma$ without a unique gate $g$ such that $\att_L(g)[0] = \gamma$ may occur; we call these \emph{open vertices}.  Circuit $D$ is then an \emph{instance} of pattern $L$ if there is a morphism $p: L \to D$ sending $\textrm{root}_L$ to $\textrm{root}_D$.  Then, given a vertex $\alpha$ in circuit $C$ and a rule $L \mapsto R$, the pair $\langle\alpha, L \mapsto R \rangle$ is a \emph{redex} if the sub-circuit of $C$ rooted at $\alpha$ (denoted $C[\alpha]$) is an instance of $L$.
\end{definition}

Most rewrite rules of $\cS$ use the open vertex labeled $\gamma$ to match ``the rest of the sub-circuit $D$.''  That is, a morphism $p$ will send $\gamma$ to each node of $D$ not explicitly mentioned to ``match'' the pattern.  Each family of rules then plays a different role in gate elimination.
\begin{description}
\item[Normalizing] enforces that each circuit does not contain ``duplicate'' gates --- such as double negation or trivial identity.
\item[Fixing] applies when a gate computes a constant function because of one input.  
\item[Passing] applies when a gate computes the identity function because of one input.
\item[Tautology] removes redundant (costly) gates which trivially compute constants. 
\end{description}

Distinct rules are required for fixing, passing, and tautology to handle left and right inputs because inclusion of ``AND is commutative'' and ``OR is commutative'' as rewrite rules makes confluence impossible \cite{Socher-Ambrosius91}.  Finally, we describe an \emph{abstract} proper rewrite step (Definition \ref{def:tgrs}) as a \emph{concrete} operation that transforms circuits (Algorithm \ref{algo:ckt-simp-system}).

\begin{figure}[t]
  \centering
  \begin{subfigure}[b]{0.23\textwidth}
    \includegraphics[width=\textwidth]{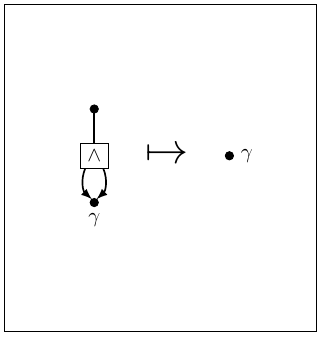}
    \caption{AND De-duplication}
  \end{subfigure}
  \begin{subfigure}[b]{0.23\textwidth}
    \centering
    \includegraphics[width=\textwidth]{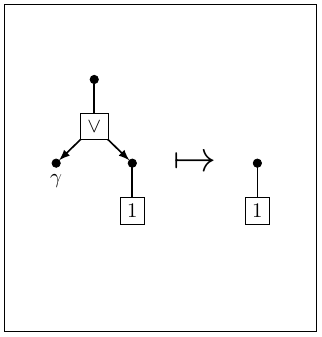}
    \caption{OR Fixing (left)}
  \end{subfigure}
  \begin{subfigure}[b]{0.23\textwidth}
    \centering
    \includegraphics[width=\textwidth]{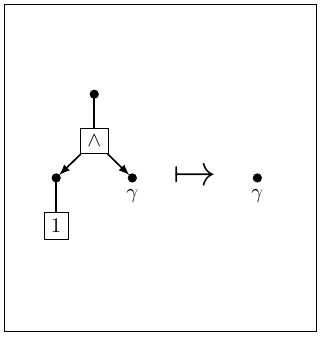}
    \caption{AND Passing (right)}
  \end{subfigure}
  \begin{subfigure}[b]{0.23\textwidth}
    \centering
    \includegraphics[width=\textwidth]{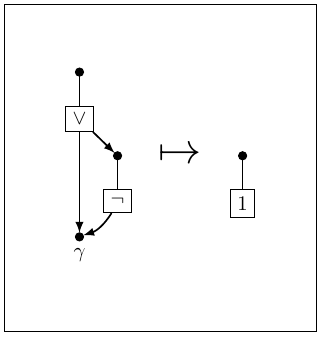}
    \caption{OR Tautology (left)}
  \end{subfigure}
  \caption{Selection of proper circuit rewriting steps, see Figures \ref{fig:normalizing-rules} -- \ref{fig:tautology-rules} in Appendix \ref{sec:gate-elim-rules} for the complete system.}
  \label{fig:selection-of-rules}
\end{figure}

\begin{algorithm}[H]
  \caption{Step of Circuit Simplification System $\cS$, defining $C \to C'$}
  \label{algo:ckt-simp-system}
  \begin{algorithmic}
    \Require $C$ is a circuit containing the redex $\langle \alpha, R \mapsto L \rangle$, with $p: R \to C[\alpha]$ the witnessing morphism
    \State $C_1 \gets C - \{ g \}$ where $g = \operatorname{res}^{-1}(\alpha)$
    \Comment{Remove the unique gate with output wire $\alpha$}
    \State $C_2 \gets C_1 + R$
    \Comment{Disjoint union: rhs of the matched rule with $C$}
    \State $C_3 \gets $ Identify vertex $\alpha$ with $\textrm{root}_R$ of $C_2$
    \Comment{Connect $R$ to the appropriate element(s) of $C$}
    \If{$\gamma \in R$}
      \Comment{Does $R$ reuse a subcircuit?}
      \State $C_4 \gets $ Identify vertex $p(\gamma)$ with $\gamma$ of $C_3$
      \Comment{Yes; wire the matching element of $R$ into $C$}
    \Else
      \State $C_4 \gets C_3$
      \Comment{$R$ does not reuse any subcircuits --- do nothing}
    \EndIf
    \State $C' \gets $ Garbage collection: remove all vertices and edges unreachable from $\mathrm{root}_{C_4}$
  \end{algorithmic}
\end{algorithm}


\section{Simplification of \texorpdfstring{$\cU_2$}{U2} and \texorpdfstring{$\cB_2$}{B2} Circuits is Not Confluent}
\label{sec:u2-ge-not-confleunt}

In this section we demonstrate that any similar formulation of gate elimination in $\cU_2$ is \textbf{\underline{not}} confluent. As $\cU_2 \subset \cB_2$, the same argument shows gate elimination in $\cB_2$ is not confluent. To do this we produce a circuit and substitution which, after applying a different series of rewrites, results in two distinct circuits (i.e. non-isomorphic). We recall the 14 binary operators that make up $\cU_2$ in Table \ref{tab:u2}.

\begin{table}[h]
\centering
\caption{The $\cU_2$ basis: all the binary Boolean operations except $\oplus$ and its negation $\odot$}
\label{tab:u2}
\begin{tabular}{|cc|cccccccccccccc|}
\hline
$p$ & $q$ & $\ast_1$ & $\ast_2$ & $\ast_3$ & $\ast_4$ & $\ast_5$ & $\ast_6$ & $\ast_7$ & $\ast_8$ & $\ast_9$ & $\ast_{10}$ & $\ast_{11}$ & $\ast_{12}$ & $\ast_{13}$ & $\ast_{14}$ \\ \hline
T & T & T     & F     & T     & F     & T     & F     & T     & F     & T     & F        & T        & F        & T        & F        \\
T & F & T     & F     & T     & F     & F     & T     & T     & F     & F     & T        & F        & T        & T        & F        \\
F & T & T     & F     & F     & T     & T     & F     & F     & T     & T     & F        & F        & T        & T        & F        \\
F & F & T     & F     & F     & T     & F     & T     & T     & F     & T     & F        & F        & T        & F        & T        \\ \hline
\end{tabular}
\end{table}

Gates $\ast_4$ and $\ast_6$ compute the negation of their first and second inputs respectively. If these are not the sole gate in the circuit, then they are superfluous. We can remove them by relabeling either their input gates or every gate that is read by them. We demonstrate the two ways to ``push'' these superfluous negations in Figure \ref{fig:pushing-example}.

\begin{figure}[H]
    \centering
    \begin{subfigure}[b]{0.31\textwidth}
       \centering
       \includegraphics[scale=0.75]{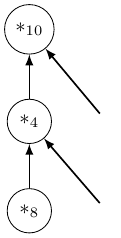}
       \caption{$\ast_4$ can be removed two ways}
    \end{subfigure}
    \begin{subfigure}[b]{0.31\textwidth}
       \centering
       \includegraphics[scale=0.75]{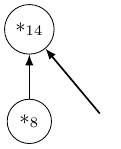}
       \caption{Pushing up: $\ast_{10} \to \ast_{14}$}
    \end{subfigure}
        \begin{subfigure}[b]{0.31\textwidth}
       \centering
       \includegraphics[scale=0.75]{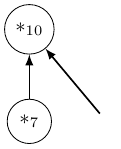}
       \caption{Pushing down: $\ast_{8} \to \ast_7$}
       \end{subfigure}
       \caption{Two simplifications for $\ast_4$, disproving convergence of gate elimination over $\cU_2$ and $\cB_2$}
       \label{fig:pushing-example}
\end{figure}

If both of these rewrite rules are present in our system then, as the example shows, is not confluent. Unfortunately, we \textbf{must} include rewrite rules that both push negations up and down: if we do not (and choose to always push negations in one direction), then superfluous $\ast_4$ and $\ast_6$ gates can survive at the extremes of the circuit. Furthermore, this issue cannot be sidestepped even if we limit ourselves to optimal circuits. This is because ``passing'' rules in this rewriting system must generate intermediate negation gates.

For example, $p \ast_8 q$ is equivalent to $\neg p \land q$; if we wanted to substitute $q \leftarrow 1$ we would like the equivalent of a ``passing'' rule to be applied. However, the output wires of $\ast_8$ should pass to $\neg p$, which is not an actual gate in the circuit. Thus, rather than removing $\ast_8$, this substitution will have to relabel $\ast_8$ to $\ast_4$. Then $\ast_4$ can be eliminated by ``pushing'' it either up or down leading to two non-isomorphic circuits. Again, we cannot always ``push down'' as $p$ could be an input and not a binary gate. From these examples, we see that

\begin{observation}\label{obs:u2-b2-not-confluent}
    Gate elimination is not confluent in the $\cU_2$ or $\cB_2$ bases, even when restricted to optimal circuits.
\end{observation}


\section{Constructive Circuit Lower Bounds via Convergent Gate Elimination}
\label{sec:rewriting-system}

Proofs by gate elimination often show a circuit $C$ has property $\cP$ with the following argument.

\begin{enumerate}
\item Assume, for the sake of contradiction, that $C$ does \textbf{not} have property $\cP$.
\item Select a variable $x_i$ and constant $\alpha$ for substitution using $\neg \cP$ and the structure of $C$.
\item \textbf{\emph{Simplify}} $C$ under the substitution $\{ x_i \mapsto \alpha \}$, to obtain a constant-free circuit $C'$.
\item Argue that a critical property $\cP'$ of $C'$ implies a contradiction.
\end{enumerate}

To demonstrate that this system is sufficiently powerful to capture gate elimination arguments, we revisit Schnorr's lower-bound for the parity function $\XOR$ using the system $\mathcal{S}$. To begin, we must first recall the parity functions $\XOR_n$:
\begin{definition}\label{def:xor}
    For an $n$ bit input $\vec{x} = (x_1, \ldots, x_n) \in \{0,1\}^n$, $\XOR_n(\vec{x}) = (\sum_{i=1}^n x_i) \bmod 2.$
\end{definition}
We define $(\neg) f$ to mean ``$f$ or $\neg f$.'' For a (partial) assignment $\alpha$, we define $f \hook_\alpha$ to be the function restricted under $\alpha$. We now recall two useful facts about $(\neg)\XOR_n$.

\begin{fact}\label{fact:xor-dsr}
    $\XOR_n$ is \emph{self-reducible}, i.e., for any $i \in [n]$, for any assignment $b \in \{0,1\},$ 
    \[\XOR_n\hook_{x_i \leftarrow b} \equiv (\neg)\XOR_{n-1}.\]
\end{fact}

\begin{fact}\label{fact:xor-non-degeneracy}
  $\XOR_n$ is \emph{maximally sensitive}, i.e., for all $i \in [n]$, for all assignments $\alpha$,
  \[\XOR_n(\alpha) \neq \XOR_n(\alpha^{\oplus_i}) \text{ where } \alpha^{\oplus_i} \text{ denotes } \alpha \text{ but with its } i\text{th bit flipped}.\]
\end{fact}

Schnorr's proof is the quintessential gate elimination proof. We provide the full details of re-proving it using our system in Appendix \ref{sec:ckt-rewrite-demo} and merely provide a sketch here.

\begin{theorem}[Schnorr, \cite{Schnorr74}] \label{thm:schnorr}
  $\XOR_n$ requires $3(n-1)$ gates in the DeMorgan basis.
\end{theorem}

\begin{proof}[Proof Sketch]
    Let $C$ be an optimal circuit for $\XOR_n$ with $n \geq 2$, our goal is to find a substitution of an input hyperedge which results in at least 3 costly gate hyperedges (those labeled with $\land$ or $\lor$) to be removed during rewriting steps. We first fix a topological order and let $h$ be its first costly gate. Then, $h$ has two arguments whose terms are $(\neg)x_i$ and $(\neg)x_j$. If $h$ were the only costly gate reading $(\neg)x_i$, we can substitute $x_j$ to eliminate $h$ via a fixing rule, leaving $x_i$ disconnected. However, the resulting circuit, which by self-reducibility (Fact \ref{fact:xor-dsr}) computes $(\neg)\XOR_{n-1}$, no longer depends on $x_i$. This contradicts non-degeneracy of $\XOR_{n-1}$ (Fact \ref{fact:xor-non-degeneracy}). 
    
    Therefore, another costly gate $f \neq h$ reads $(\neg)x_i$. Furthermore, $(\neg)f$ cannot be the output since we can fix it by substituting $x_i$, again contradicting Fact \ref{fact:xor-dsr}. Let $f'$ be any binary gate reading $f$. Then, we choose the substitution for $x_i$ that fixes $f$. This ultimately removes at least 3 costly gates $h, f, f'$ after simplification and the resulting circuit computes $\XOR_{n-1}$ by Fact \ref{fact:xor-dsr}. Finally, we apply induction to obtain the desired lower bound.
\end{proof}

Schnorr's proof is simple, but when viewed through the lens of formal methods, it is hard not to see it as an algorithm for breaking down $\XOR$ circuits; this is especially true of our structured proof in Appendix \ref{sec:schnorr-structured-proof}. Gate elimination yields its lower bounds by arguing that if a circuit does not have some minimal structure, then it doesn't compute its underlying function. Since a function is defined by its behavior on its inputs, it seems reasonable to expect that we can derive which inputs the circuit would fail on. We show that this intuition is true at least for Schnorr's proof by giving the first constructive separation for a circuit lower bound proved by gate elimination. 

\begin{theorem}
  \label{thm:xor-refuter}
    Let $C$ be a DeMorgan circuit on $n$ inputs with size $s < 3(n - 1)$ that claims to compute $\XOR_n$. We can construct an input $x \in \{0, 1\}^n$ such that $C(x) \neq \XOR_n(x)$ in polynomial time with respect to $n$.
\end{theorem}

\begin{proof}[Proof]
    Let $C$ be a DeMorgan circuit on $n$ inputs with size $s < 3(n - 1)$ and let $f$ be the function it computes. If $n = 2$, then we can brute-force evaluate all four possible inputs to find a counterexample in constant time. For larger $n$, we can run Algorithm \ref{algo:schnorr}. At a high level, the algorithm will follow Schnorr's proof; each iteration, it tries to find a substitution that removes at least three gates. If it cannot, then the circuit's function will violate Fact \ref{fact:xor-non-degeneracy}, and we can easily extract a counterexample. If successful, it instead produces a circuit that is too small to compute $\XOR$ on one fewer variable which it can recurse on. This continues until, eventually, the remaining circuit purports to compute $(\neg)\XOR_2$ with fewer than three gates, and again we can simply brute-force over its four completions.

    Algorithm \ref{algo:schnorr} keeps track of (i) a partial restriction $\alpha$ and (ii) a list $V$ of the variables still unrestricted by the circuit. During each iteration, it considers $C_i$, a \textbf{normalized} circuit computing $f\hook_{\alpha}$. In lines 4-5, the algorithm can identify an obvious contradiction to Fact \ref{fact:xor-non-degeneracy}. If $C_i$ is degenerate with respect to some variable(s) in $V$ then $f\hook_{\alpha}$ does not depend on some $x_j$ for some $j \in V$. Let $\hat{\alpha}$ be any completion of $\alpha$, then $C (\hat{\alpha}) = f(\hat{\alpha}) = f(\hat{\alpha}^{\oplus_j}) = C(\hat{\alpha}^{\oplus_j})$; $C$ is incorrect on one of these two inputs since $\XOR_n(\hat{\alpha}) \neq \XOR_n(\hat{\alpha}^{\oplus_j})$.

    Otherwise, in line 7, the algorithm selects the topologically minimal binary gate of $C_i$ which must read $(\neg)x_p$ and $(\neg)x_q$ where $i \neq j$ since $C_i$ is normalized. As in Schnorr, $(\neg)x_p$ must be read by another binary gate $f$ besides $h$. If this is not the case on line 9 then we are in a similar case to the previous $\texttt{DEGEN}$ case: after adding the correct substitution of $x_q$ to $\alpha$ and normalizing, $f\hook_\alpha$ will no longer depend on $x_p$, and $C$ must err on any completion $\hat{\alpha}$ or $\hat{\alpha}^{\oplus_p}.$ Furthermore, $(\neg)f$ must not be the output or else we can fix the circuit to be constant on line 15-17. If it were, we could then take any completion $\hat{\alpha}$ and it or $\hat{\alpha}^{\oplus_q}$ will be a counterexample, since $f$ no longer depends on $x_q$.

    If the algorithm does not return prematurely, then as in Schnorr, there exists some $f'$ that reads $(\neg)f$. Algorithm \ref{algo:schnorr} selects a substitution for $x_p$ so that simplification removes $h, f$ and $f'$ reducing the size of the circuit by at least three per loop iteration. Hence, the resulting $C_{i+1}$ is too small to compute $\XOR_{|V|}$. In the worst case, the algorithm goes through $n-2$ iterations and the resulting $C_{n-2}$ will, by Fact \ref{fact:xor-dsr} purportedly compute $(\neg)\XOR_2$. However, its size will be at most $s - 3(n-2) < 3(n-1) - 3(n-2) = 3$. Therefore, we can simply brute-force over all four completions of $\alpha$ and $C$ must err on at least one of them.

    It is easy to see this process runs in polynomial time (with respect to $n$). There are at most linear loop iterations, and each condition can be checked in polynomial time. Since we are using system $\mathcal{S}'$ without collapse, substitution and simplification in each iteration can be done in polynomial time. To extract a counterexample, we only need to compute the circuit on a constant number of inputs and each evaluation can be done in polynomial time. 
\end{proof}

\begin{algorithm}[h]
  \caption{Efficient search for a \emph{bad restriction} $\alpha$ witnessing that $C$ does not compute $\XOR$.}
  \label{algo:schnorr}
  \begin{algorithmic}[1]
    \Require $C$ is a circuit on $n$ input variables with $s$ gates with $n > 3$ and $s < 3(n-1)$
    \State $i \gets 0 $ ~;~ $C_0 \gets \text{Normalize } C$ ~;~ $\alpha \gets \emptyset$ ~;~ $V \gets [n]$
    \LComment{number of rounds ~;~ working circuit ~;~ working restriction ~;~ active variables}
    \While{|V| > 2}
    \If{a variable $x_j$ for some $j \in V$ is not read in $C_i$}
    \State \Return "\texttt{DEGEN} in $x_j$ under $\alpha$"
    \EndIf
    \State \textbf{Assert:} $C_i$ is in normal form and does not trivially fail to compute $\XOR_{|V|}$
    \State $h \gets$ topologically minimal $\{\land, \lor\}$-gate in $C_i$
    \State \textbf{Assert:} $h$ has inputs $(\neg)x_p, (\neg)x_q$ for $p, q \in V$ and $p \neq q$ \LComment{the assertion holds because $C_i$ is in normal form}
    \If{the fanout of $(\neg)x_p$ is 1}
    \State add $\mathsf{fixer}(h,x_q) \mapsto x_q$ to $\alpha$
    \LComment{$\mathsf{fixer}(h, x_q)$ is the substitution for $x_q$ that removes $h$ via a fixing rule}
    \State \Return "\texttt{DEGEN} in $x_p$ under $\alpha$"
    \Else
    \State $f \gets$ the distinct gate reading $x_p$ besides $h$, so $f \neq h$
    \EndIf
    \If{$(\neg)f$ is the output gate of $C_i$}
    \State add $\mathsf{fixer}(f, x_p) \mapsto x_p$ to $\alpha$
    \State \Return "\texttt{CONST} under $\alpha$"
    \Else
    \State $f' \gets$ topologically minimal costly successor of $f$ in $C_i$
    \EndIf
    \State add $\mathsf{fixer}(f,x_p) \mapsto x_p$ to $\alpha$
    \State $C_{i+1} \gets $ Substitute $x_p \gets \mathsf{fixer}(f,x_p)$ in $C_i$ and simplify
    \Comment{Eliminates $h, f,$ and $f'$}
    \LComment{By convergence of system $\cS'$ this step is deterministic --- greedy selection of simplifications is the same up to bisimulation of $C_{i+1}$ which is all we require }
    \State $V \gets V \setminus \{ p \}$
    \State $i \gets (i + 1)$
    \State \textbf{Assert:} $C_{i+1}$ is a normalized circuit which is too small to compute $(\neg)\XOR_{|V|}$ 
    \EndWhile
    \State \Return "\texttt{FAILS} under $\alpha$"
    \Comment{``fail search'', brute-force is feasible to find the error}
  \end{algorithmic}
\end{algorithm}


\clearpage



\bibliography{references}

@ARTICLE{Shannon49,
    author={Shannon, Claude. E.},
    journal={The Bell System Technical Journal}, 
    title={The synthesis of two-terminal switching circuits}, 
    year={1949},
    volume={28},
    number={1},
    pages={59-98},
    doi={10.1002/j.1538-7305.1949.tb03624.x}
}

@incollection{KNUTH1970,
    title = {Simple Word Problems in Universal Algebras††The work reported in this paper was supported in part by the U.S. Office of Naval Research.},
    editor = {JOHN LEECH},
    booktitle = {Computational Problems in Abstract Algebra},
    publisher = {Pergamon},
    pages = {263-297},
    year = {1970},
    isbn = {978-0-08-012975-4},
    doi = {https://doi.org/10.1016/B978-0-08-012975-4.50028-X},
    url = {https://www.sciencedirect.com/science/article/pii/B978008012975450028X},
    author = {DONALD E. KNUTH and PETER B. BENDIX}
}

@article{Stockmeyer77,
    author    = {Larry J. Stockmeyer},
    title     = {On the Combinational Complexity of Certain Symmetric Boolean Functions},
    journal   = {Math. Syst. Theory},
    volume    = {10},
    pages     = {323--336},
    year      = {1977},
    url       = {https://doi.org/10.1007/BF01683282},
    doi       = {10.1007/BF01683282},
    bibsource = {dblp computer science bibliography, https://dblp.org}
}

@article{HUET1981,
    title = {A complete proof of correctness of the Knuth-Bendix completion algorithm},
    journal = {Journal of Computer and System Sciences},
    volume = {23},
    number = {1},
    pages = {11-21},
    year = {1981},
    issn = {0022-0000},
    doi = {https://doi.org/10.1016/0022-0000(81)90002-7},
    url = {https://www.sciencedirect.com/science/article/pii/0022000081900027},
    author = {Gérard Huet}
}

@article{Redkin1973,
    title={Proof of minimality of circuits consisting of functional elements},
    author={Red’kin, NP},
    journal={Systems Theory Research: Problemy Kibernetiki},
    pages={85--103},
    year={1973},
    publisher={Springer}
}

@article{Schnorr74,
    author    = {Claus{-}Peter Schnorr},
    title     = {Zwei lineare untere Schranken f{\"{u}}r die Komplexit{\"{a}}t
               Boolescher Funktionen},
    journal   = {Computing},
    volume    = {13},
    number    = {2},
    pages     = {155--171},
    year      = {1974},
    url       = {https://doi.org/10.1007/BF02246615},
    doi       = {10.1007/BF02246615},
    bibsource = {dblp computer science bibliography, https://dblp.org}
}

@book{Wegener1987,
    place     = {Stuttgart},
    title     = {The Complexity of Boolean Functions},
    publisher = {Wiley-Teubner},
    author    = {Wegener, Ingo},
    year      = {1987}
}

@inproceedings{Socher-Ambrosius91,
  author       = {Rolf Socher{-}Ambrosius},
  editor       = {Ronald V. Book},
  title        = {Boolean Algebra Admits No Convergent Term Rewriting System},
  booktitle    = {Rewriting Techniques and Applications, 4th International Conference,
                  RTA-91, Como, Italy, April 10-12, 1991, Proceedings},
  series       = {Lecture Notes in Computer Science},
  volume       = {488},
  pages        = {264--274},
  publisher    = {Springer},
  year         = {1991},
  url          = {https://doi.org/10.1007/3-540-53904-2\_102},
  doi          = {10.1007/3-540-53904-2\_102},
  bibsource    = {dblp computer science bibliography, https://dblp.org}
}

@article{Zwick91,
  title={A 4n lower bound on the combinational complexity of certain symmetric boolean functions over the basis of unate dyadic Boolean functions},
  author={Zwick, Uri},
  journal={SIAM Journal on Computing},
  volume={20},
  number={3},
  pages={499--505},
  year={1991},
  publisher={SIAM}
}

@inbook{Plaisted93,
    author = {Plaisted, David A.},
    title = {Equational reasoning and term rewriting systems},
    year = {1993},
    isbn = {019853745X},
    publisher = {Oxford University Press, Inc.},
    address = {USA},
    booktitle = {Handbook of Logic in Artificial Intelligence and Logic Programming (Vol. 1)},
    pages = {274–364},
    numpages = {91}
}

@article{lamport1995write,
  title={How to write a proof},
  author={Lamport, Leslie},
  journal={The American mathematical monthly},
  volume={102},
  number={7},
  pages={600--608},
  year={1995},
  publisher={Taylor \& Francis}
}

@Book{BaaderN1998,
  author       = {Franz Baader and Tobias Nipkow},
  title        = {Term rewriting and all that},
  year         = 1998,
  publisher    = {Cambridge University Press},
  isbn         = {978-0-521-45520-6},
  bibsource    = {dblp computer science bibliography,
                  https://dblp.org}
}

@article{Plump99,
  title={Term graph rewriting},
  author={Plump, Detlef},
  journal={Handbook Of Graph Grammars And Computing By Graph Transformation: Volume 2: Applications, Languages and Tools},
  pages={3--61},
  year={1999},
  publisher={World Scientific}
}

@inproceedings{IwamaM02,
    author    = {Kazuo Iwama and
               Hiroki Morizumi},
    editor    = {Krzysztof Diks and
               Wojciech Rytter},
    title     = {An Explicit Lower Bound of 5n - o(n) for Boolean Circuits},
    booktitle = {Mathematical Foundations of Computer Science 2002, 27th International
               Symposium, {MFCS} 2002, Warsaw, Poland, August 26-30, 2002, Proceedings},
    series    = {Lecture Notes in Computer Science},
    volume    = {2420},
    pages     = {353--364},
    publisher = {Springer},
    year      = {2002},
    url       = {https://doi.org/10.1007/3-540-45687-2\_29},
    doi       = {10.1007/3-540-45687-2\_29},
    bibsource = {dblp computer science bibliography, https://dblp.org}
}

@inproceedings{LachishR01,
    author = {Lachish, Oded and Raz, Ran},
    title = {Explicit lower bound of 4.5n - o(n) for boolena circuits},
    year = {2001},
    isbn = {1581133499},
    publisher = {Association for Computing Machinery},
    address = {New York, NY, USA},
    url = {https://doi.org/10.1145/380752.380832},
    doi = {10.1145/380752.380832},
    booktitle = {Proceedings of the Thirty-Third Annual ACM Symposium on Theory of Computing},
    pages = {399–408},
    numpages = {10},
    location = {Hersonissos, Greece},
    series = {STOC '01}
}

@article{Kombarov2011,
    title={The minimal circuits for linear Boolean functions},
    author={Kombarov, Yu A},
    journal={Moscow University Mathematics Bulletin},
    volume={66},
    number={6},
    pages={260--263},
    year={2011},
    publisher={Springer}
}

@inproceedings{DemenkovK11,
    author    = {Evgeny Demenkov and
               Alexander S. Kulikov},
    editor    = {Filip Murlak and
               Piotr Sankowski},
    title     = {An Elementary Proof of a 3n - o(n) Lower Bound on the Circuit Complexity
               of Affine Dispersers},
    booktitle = {Mathematical Foundations of Computer Science 2011 - 36th International
               Symposium, {MFCS} 2011, Warsaw, Poland, August 22-26, 2011. Proceedings},
    series    = {Lecture Notes in Computer Science},
    volume    = {6907},
    pages     = {256--265},
    publisher = {Springer},
    year      = {2011},
    url       = {https://doi.org/10.1007/978-3-642-22993-0\_25},
    doi       = {10.1007/978-3-642-22993-0\_25},
    bibsource = {dblp computer science bibliography, https://dblp.org}
}

@inproceedings{SternagelZ12,
    author       = {Thomas Sternagel and
                  Harald Zankl},
    editor       = {Bernhard Gramlich and
                  Dale Miller and
                  Uli Sattler},
    title        = {{KBCV} - Knuth-Bendix Completion Visualizer},
    booktitle    = {Automated Reasoning - 6th International Joint Conference, {IJCAR}
                  2012, Manchester, UK, June 26-29, 2012. Proceedings},
    series       = {Lecture Notes in Computer Science},
    volume       = {7364},
    pages        = {530--536},
    publisher    = {Springer},
    year         = {2012},
    url          = {https://doi.org/10.1007/978-3-642-31365-3\_41},
    doi          = {10.1007/978-3-642-31365-3\_41},
    bibsource    = {dblp computer science bibliography, https://dblp.org}
}

@article{lamport2012write,
    title={How to write a 21 st century proof},
    author={Lamport, Leslie},
    journal={Journal of fixed point theory and applications},
    volume={11},
    pages={43--63},
    year={2012},
    publisher={Springer}
}

@INPROCEEDINGS{FindGHK2016,
    author={Find, Magnus Gausdal and Golovnev, Alexander and Hirsch, Edward A. and Kulikov, Alexander S.},
    booktitle={2016 IEEE 57th Annual Symposium on Foundations of Computer Science (FOCS)}, 
    title={A Better-Than-3n Lower Bound for the Circuit Complexity of an Explicit Function}, 
    year={2016},
    volume={},
    number={},
    pages={89-98},
    doi={10.1109/FOCS.2016.19}
}

@article{GolovnevHKK18,
    author    = {Alexander Golovnev and
               Edward A. Hirsch and
               Alexander Knop and
               Alexander S. Kulikov},
    title     = {On the limits of gate elimination},
    journal   = {J. Comput. Syst. Sci.},
    volume    = {96},
    pages     = {107--119},
    year      = {2018},
    url       = {https://doi.org/10.1016/j.jcss.2018.04.005},
    doi       = {10.1016/j.jcss.2018.04.005},
    bibsource = {dblp computer science bibliography, https://dblp.org}
}

@article{Kombarov2018,
    title={Complexity and Structure of Circuits for Parity Functions},
    author={Kombarov, Yu A},
    journal={Journal of Mathematical Sciences},
    volume={233},
    pages={95--99},
    year={2018},
    publisher={Springer}
}

@article{Ilango20,
    author = {Ilango, Rahul},
    title = {Constant Depth Formula and Partial Function Versions of MCSP Are Hard},
    journal = {SIAM Journal on Computing},
    volume = {0},
    number = {0},
    pages = {FOCS20-317-FOCS20-367},
    year = {2020},
    doi = {10.1137/20M1383562},
    
    URL = {    
            https://doi.org/10.1137/20M1383562
        },
    eprint = {    
            https://doi.org/10.1137/20M1383562
        }
}

@inproceedings{Li022,
    author    = {Jiatu Li and
               Tianqi Yang},
    editor    = {Stefano Leonardi and
               Anupam Gupta},
    title     = {3.1\emph{n} - \emph{o}(\emph{n}) circuit lower bounds for explicit
               functions},
    booktitle = {{STOC} '22: 54th Annual {ACM} {SIGACT} Symposium on Theory of Computing,
               Rome, Italy, June 20 - 24, 2022},
    pages     = {1180--1193},
    publisher = {{ACM}},
    year      = {2022},
    url       = {https://doi.org/10.1145/3519935.3519976},
    doi       = {10.1145/3519935.3519976},
    bibsource = {dblp computer science bibliography, https://dblp.org}
}

@inproceedings{WebbHU23,
  author       = {Brae J. Webb and
                  Ian J. Hayes and
                  Mark Utting},
  editor       = {Robbert Krebbers and
                  Dmitriy Traytel and
                  Brigitte Pientka and
                  Steve Zdancewic},
  title        = {Verifying Term Graph Optimizations using {Isabelle/HOL}},
  booktitle    = {Proceedings of the 12th {ACM} {SIGPLAN} International Conference on
                  Certified Programs and Proofs, {CPP} 2023, Boston, MA, USA, January
                  16-17, 2023},
  pages        = {320--333},
  publisher    = {{ACM}},
  year         = {2023},
  url          = {https://doi.org/10.1145/3573105.3575673},
  doi          = {10.1145/3573105.3575673},
  bibsource    = {dblp computer science bibliography, https://dblp.org}
}

@inproceedings{Paul75,
author = {Paul, Wolfgang J.},
title = {A 2.5 n-lower bound on the combinational complexity of Boolean functions},
year = {1975},
isbn = {9781450374194},
publisher = {Association for Computing Machinery},
address = {New York, NY, USA},
url = {https://doi.org/10.1145/800116.803750},
doi = {10.1145/800116.803750},
booktitle = {Proceedings of the Seventh Annual ACM Symposium on Theory of Computing},
pages = {27–36},
numpages = {10},
location = {Albuquerque, New Mexico, USA},
series = {STOC '75}
}

@inproceedings{Sattler81,
  author       = {J{\"{u}}rgen Sattler},
  editor       = {Peter Deussen},
  title        = {Netzwerke zur simultanen Berechnung Boolescher Funktionen},
  booktitle    = {Theoretical Computer Science, 5th GI-Conference, Karlsruhe, Germany,
                  March 23-25, 1981, Proceedings},
  series       = {Lecture Notes in Computer Science},
  volume       = {104},
  pages        = {32--40},
  publisher    = {Springer},
  year         = {1981},
  url          = {https://doi.org/10.1007/BFb0017293},
  doi          = {10.1007/BFB0017293},
  bibsource    = {dblp computer science bibliography, https://dblp.org}
}

@article{BlumS84,
  author       = {Norbert Blum and
                  Martin Seysen},
  title        = {Characterization of all Optimal Networks for a Simultaneous Computation
                  of {AND} and {NOR}},
  journal      = {Acta Informatica},
  volume       = {21},
  pages        = {171--181},
  year         = {1984},
  url          = {https://doi.org/10.1007/BF00289238},
  doi          = {10.1007/BF00289238},
  bibsource    = {dblp computer science bibliography, https://dblp.org}
}

@inproceedings{AverkovBEGKKKLL25,
  author       = {Daniil Averkov and
                  Tatiana Belova and
                  Gregory Emdin and
                  Mikhail Goncharov and
                  Viktoriia Krivogornitsyna and
                  Alexander S. Kulikov and
                  Fedor Kurmazov and
                  Daniil Levtsov and
                  Georgie Levtsov and
                  Vsevolod Vaskin and
                  Aleksey Vorobiev},
  editor       = {Toby Walsh and
                  Julie Shah and
                  Zico Kolter},
  title        = {Cirbo: {A} New Tool for Boolean Circuit Analysis and Synthesis},
  booktitle    = {AAAI-25, Sponsored by the Association for the Advancement of Artificial
                  Intelligence, February 25 - March 4, 2025, Philadelphia, PA, {USA}},
  pages        = {11105--11112},
  publisher    = {{AAAI} Press},
  year         = {2025},
  url          = {https://doi.org/10.1609/aaai.v39i11.33207},
  doi          = {10.1609/AAAI.V39I11.33207},
  bibsource    = {dblp computer science bibliography, https://dblp.org}
}

@article{CarmosinoDJ2025,
  title={Simple Circuit Extensions for XOR in PTIME},
  author={Carmosino, Marco and Dang, Ngu and Jackman, Tim},
  journal={arXiv preprint arXiv:2511.16903. To Appear in STACS 2026},
  year={2025}
}

@article{GoncharovKL2025,
  author       = {Mikhail Goncharov and Alexander S. Kulikov and Georgie Levtsov},
  title        = {Smaller Circuits for Bit Addition},
  journal={arXiv preprint arXiv:2509.13966. To Appear in STACS 2026},
  year={2025}
}

@article{Thiemann2022,
  author  = {René Thiemann},
  title   = {Clique is not solvable by monotone circuits of polynomial size},
  journal = {Archive of Formal Proofs},
  month   = {May},
  year    = {2022},
  note    = {\url{https://isa-afp.org/entries/Clique_and_Monotone_Circuits.html},
             Formal proof development},
  ISSN    = {2150-914x},
}

@inproceedings{BraytonM10,
  author       = {Robert K. Brayton and
                  Alan Mishchenko},
  editor       = {Tayssir Touili and
                  Byron Cook and
                  Paul B. Jackson},
  title        = {{ABC:} An Academic Industrial-Strength Verification Tool},
  booktitle    = {Computer Aided Verification, 22nd International Conference, {CAV}
                  2010, Edinburgh, UK, July 15-19, 2010. Proceedings},
  series       = {Lecture Notes in Computer Science},
  volume       = {6174},
  pages        = {24--40},
  publisher    = {Springer},
  year         = {2010},
  url          = {https://doi.org/10.1007/978-3-642-14295-6\_5},
  doi          = {10.1007/978-3-642-14295-6\_5},
  bibsource    = {dblp computer science bibliography, https://dblp.org}
}

@article{SoekenRHM18,
  author       = {Mathias Soeken and
                  Heinz Riener and
                  Winston Haaswijk and
                  Giovanni De Micheli},
  title        = {The {EPFL} Logic Synthesis Libraries},
  journal      = {CoRR},
  volume       = {abs/1805.05121},
  year         = {2018},
  url          = {http://arxiv.org/abs/1805.05121},
  eprinttype    = {arXiv},
  eprint       = {1805.05121},
  bibsource    = {dblp computer science bibliography, https://dblp.org}
}

@inproceedings{ReichlSS23,
  author       = {Franz{-}Xaver Reichl and
                  Friedrich Slivovsky and
                  Stefan Szeider},
  editor       = {Brian Williams and
                  Yiling Chen and
                  Jennifer Neville},
  title        = {Circuit Minimization with QBF-Based Exact Synthesis},
  booktitle    = {Thirty-Seventh {AAAI} Conference on Artificial Intelligence, {AAAI}
                  2023, Thirty-Fifth Conference on Innovative Applications of Artificial
                  Intelligence, {IAAI} 2023, Thirteenth Symposium on Educational Advances
                  in Artificial Intelligence, {EAAI} 2023, Washington, DC, USA, February
                  7-14, 2023},
  pages        = {4087--4094},
  publisher    = {{AAAI} Press},
  year         = {2023},
  url          = {https://doi.org/10.1609/aaai.v37i4.25524},
  doi          = {10.1609/AAAI.V37I4.25524},
  bibsource    = {dblp computer science bibliography, https://dblp.org}
}

@inproceedings{TangZZCLYX25,
  author       = {Ruofei Tang and
                  Xuliang Zhu and
                  Xinyi Zhang and
                  Lei Chen and
                  Xing Li and
                  Mingxuan Yuan and
                  Jianliang Xu},
  title        = {{EDGE:} DBMS-Empowered Boolean Decomposition for {GIG} Synthesis},
  booktitle    = {62nd {ACM/IEEE} Design Automation Conference, {DAC} 2025, San Francisco,
                  CA, USA, June 22-25, 2025},
  pages        = {1--7},
  publisher    = {{IEEE}},
  year         = {2025},
  url          = {https://doi.org/10.1109/DAC63849.2025.11133306},
  doi          = {10.1109/DAC63849.2025.11133306},
  bibsource    = {dblp computer science bibliography, https://dblp.org}
}

@article{ShirazH18,
  author       = {Sumayya Shiraz and
                  Osman Hasan},
  title        = {A Library for Combinational Circuit Verification Using the {HOL} Theorem
                  Prover},
  journal      = {{IEEE} Trans. Comput. Aided Des. Integr. Circuits Syst.},
  volume       = {37},
  number       = {2},
  pages        = {512--516},
  year         = {2018},
  url          = {https://doi.org/10.1109/TCAD.2017.2705049},
  doi          = {10.1109/TCAD.2017.2705049},
  bibsource    = {dblp computer science bibliography, https://dblp.org}
}

@inproceedings{Braibant11,
  author       = {Thomas Braibant},
  editor       = {Jean{-}Pierre Jouannaud and
                  Zhong Shao},
  title        = {Coquet: {A} Coq Library for Verifying Hardware},
  booktitle    = {Certified Programs and Proofs - First International Conference, {CPP}
                  2011, Kenting, Taiwan, December 7-9, 2011. Proceedings},
  series       = {Lecture Notes in Computer Science},
  volume       = {7086},
  pages        = {330--345},
  publisher    = {Springer},
  year         = {2011},
  url          = {https://doi.org/10.1007/978-3-642-25379-9\_24},
  doi          = {10.1007/978-3-642-25379-9\_24},
  bibsource    = {dblp computer science bibliography, https://dblp.org}
}

@inproceedings{RandPZ17,
  author       = {Robert Rand and
                  Jennifer Paykin and
                  Steve Zdancewic},
  editor       = {Bob Coecke and
                  Aleks Kissinger},
  title        = {{QWIRE} Practice: Formal Verification of Quantum Circuits in Coq},
  booktitle    = {Proceedings 14th International Conference on Quantum Physics and Logic,
                  {QPL} 2017, Nijmegen, The Netherlands, 3-7 July 2017},
  series       = {{EPTCS}},
  volume       = {266},
  pages        = {119--132},
  year         = {2017},
  url          = {https://doi.org/10.4204/EPTCS.266.8},
  doi          = {10.4204/EPTCS.266.8},
  bibsource    = {dblp computer science bibliography, https://dblp.org}
}

@article{DBLP:journals/jcss/NisanW94,
  author       = {Noam Nisan and
                  Avi Wigderson},
  title        = {Hardness vs Randomness},
  journal      = {J. Comput. Syst. Sci.},
  volume       = {49},
  number       = {2},
  pages        = {149--167},
  year         = {1994},
  url          = {https://doi.org/10.1016/S0022-0000(05)80043-1},
  doi          = {10.1016/S0022-0000(05)80043-1},
  bibsource    = {dblp computer science bibliography, https://dblp.org}
}

@inproceedings{DBLP:conf/stoc/ImpagliazzoW97,
  author       = {Russell Impagliazzo and
                  Avi Wigderson},
  editor       = {Frank Thomson Leighton and
                  Peter W. Shor},
  title        = {\emph{P = BPP} if \emph{E} Requires Exponential Circuits: Derandomizing
                  the {XOR} Lemma},
  booktitle    = {Proceedings of the Twenty-Ninth Annual {ACM} Symposium on the Theory
                  of Computing, El Paso, Texas, USA, May 4-6, 1997},
  pages        = {220--229},
  publisher    = {{ACM}},
  year         = {1997},
  url          = {https://doi.org/10.1145/258533.258590},
  doi          = {10.1145/258533.258590},
  bibsource    = {dblp computer science bibliography, https://dblp.org}
}

@article{DBLP:journals/cc/KabanetsI04,
  author       = {Valentine Kabanets and
                  Russell Impagliazzo},
  title        = {Derandomizing Polynomial Identity Tests Means Proving Circuit Lower
                  Bounds},
  journal      = {Comput. Complex.},
  volume       = {13},
  number       = {1-2},
  pages        = {1--46},
  year         = {2004},
  url          = {https://doi.org/10.1007/s00037-004-0182-6},
  doi          = {10.1007/S00037-004-0182-6},
  bibsource    = {dblp computer science bibliography, https://dblp.org}
}

@article{DBLP:journals/mst/FurstSS84,
  author       = {Merrick L. Furst and
                  James B. Saxe and
                  Michael Sipser},
  title        = {Parity, Circuits, and the Polynomial-Time Hierarchy},
  journal      = {Math. Syst. Theory},
  volume       = {17},
  number       = {1},
  pages        = {13--27},
  year         = {1984},
  url          = {https://doi.org/10.1007/BF01744431},
  doi          = {10.1007/BF01744431},
  bibsource    = {dblp computer science bibliography, https://dblp.org}
}

@inproceedings{DBLP:conf/stoc/Smolensky87,
  author       = {Roman Smolensky},
  editor       = {Alfred V. Aho},
  title        = {Algebraic Methods in the Theory of Lower Bounds for Boolean Circuit
                  Complexity},
  booktitle    = {Proceedings of the 19th Annual {ACM} Symposium on Theory of Computing,
                  1987, New York, New York, {USA}},
  pages        = {77--82},
  publisher    = {{ACM}},
  year         = {1987},
  url          = {https://doi.org/10.1145/28395.28404},
  doi          = {10.1145/28395.28404},
  bibsource    = {dblp computer science bibliography, https://dblp.org}
}

@article{DBLP:journals/jacm/Williams14,
  author       = {Ryan Williams},
  title        = {Nonuniform {ACC} Circuit Lower Bounds},
  journal      = {J. {ACM}},
  volume       = {61},
  number       = {1},
  pages        = {2:1--2:32},
  year         = {2014},
  url          = {https://doi.org/10.1145/2559903},
  doi          = {10.1145/2559903},
  bibsource    = {dblp computer science bibliography, https://dblp.org}
}

@article{DBLP:journals/theoretics/00010S024,
  author       = {Lijie Chen and
                  Ce Jin and
                  Rahul Santhanam and
                  Ryan Williams},
  title        = {Constructive Separations and Their Consequences},
  journal      = {TheoretiCS},
  volume       = {3},
  year         = {2024},
  url          = {https://doi.org/10.46298/theoretics.24.3},
  doi          = {10.46298/THEORETICS.24.3},
  bibsource    = {dblp computer science bibliography, https://dblp.org}
}

@inproceedings{DBLP:conf/stoc/GrosserC25,
  author       = {Stefan Grosser and
                  Marco Carmosino},
  editor       = {Michal Kouck{\'{y}} and
                  Nikhil Bansal},
  title        = {Student-Teacher Constructive Separations and (Un)Provability in Bounded
                  Arithmetic: Witnessing the Gap},
  booktitle    = {Proceedings of the 57th Annual {ACM} Symposium on Theory of Computing,
                  {STOC} 2025, Prague, Czechia, June 23-27, 2025},
  pages        = {1341--1347},
  publisher    = {{ACM}},
  year         = {2025},
  url          = {https://doi.org/10.1145/3717823.3718216},
  doi          = {10.1145/3717823.3718216},
  bibsource    = {dblp computer science bibliography, https://dblp.org}
}

@inproceedings{SternagelT13,
  author       = {Christian Sternagel and
                  Ren{\'{e}} Thiemann},
  editor       = {Femke van Raamsdonk},
  title        = {Formalizing Knuth-Bendix Orders and Knuth-Bendix Completion},
  booktitle    = {24th International Conference on Rewriting Techniques and Applications,
                  {RTA} 2013, Eindhoven, The Netherlands, June 24-26, 2013},
  series       = {LIPIcs},
  volume       = {21},
  pages        = {287--302},
  publisher    = {Schloss Dagstuhl - Leibniz-Zentrum f{\"{u}}r Informatik},
  year         = {2013},
  url          = {https://doi.org/10.4230/LIPIcs.RTA.2013.287},
  doi          = {10.4230/LIPICS.RTA.2013.287},
  bibsource    = {dblp computer science bibliography, https://dblp.org}
}

@inproceedings{DBLP:conf/stoc/Maass84,
  author       = {Wolfgang Maass},
  editor       = {Richard A. DeMillo},
  title        = {Quadratic Lower Bounds for Deterministic and Nondeterministic One-Tape
                  Turing Machines (Extended Abstract)},
  booktitle    = {Proceedings of the 16th Annual {ACM} Symposium on Theory of Computing,
                  April 30 - May 2, 1984, Washington, DC, {USA}},
  pages        = {401--408},
  publisher    = {{ACM}},
  year         = {1984},
  url          = {https://doi.org/10.1145/800057.808706},
  doi          = {10.1145/800057.808706},
  bibsource    = {dblp computer science bibliography, https://dblp.org}
}

@inproceedings{DBLP:conf/focs/0001LO24,
  author       = {Lijie Chen and
                  Jiatu Li and
                  Igor C. Oliveira},
  title        = {Reverse Mathematics of Complexity Lower Bounds},
  booktitle    = {65th {IEEE} Annual Symposium on Foundations of Computer Science, {FOCS}
                  2024, Chicago, IL, USA, October 27-30, 2024},
  pages        = {505--527},
  publisher    = {{IEEE}},
  year         = {2024},
  url          = {https://doi.org/10.1109/FOCS61266.2024.00040},
  doi          = {10.1109/FOCS61266.2024.00040},
  bibsource    = {dblp computer science bibliography, https://dblp.org}
}

@article{DBLP:journals/sigact/Oliveira25,
  author       = {Igor C. Oliveira},
  title        = {{SIGACT} News Complexity Theory Column 124 Meta-Mathematics of Computational
                  Complexity Theory},
  journal      = {{SIGACT} News},
  volume       = {56},
  number       = {1},
  pages        = {41--68},
  year         = {2025},
  url          = {https://doi.org/10.1145/3726856.3726862},
  doi          = {10.1145/3726856.3726862},
  bibsource    = {dblp computer science bibliography, https://dblp.org}
}

@inproceedings{DBLP:conf/stacs/RenS22,
  author       = {Hanlin Ren and
                  Rahul Santhanam},
  editor       = {Petra Berenbrink and
                  Benjamin Monmege},
  title        = {A Relativization Perspective on Meta-Complexity},
  booktitle    = {39th International Symposium on Theoretical Aspects of Computer Science,
                  {STACS} 2022, Marseille, France (Virtual Conference), March 15-18,
                  2022},
  series       = {LIPIcs},
  volume       = {219},
  pages        = {54:1--54:13},
  publisher    = {Schloss Dagstuhl - Leibniz-Zentrum f{\"{u}}r Informatik},
  year         = {2022},
  url          = {https://doi.org/10.4230/LIPIcs.STACS.2022.54},
  doi          = {10.4230/LIPICS.STACS.2022.54},
  bibsource    = {dblp computer science bibliography, https://dblp.org}
}


\clearpage
\appendix
\section{Gate Elimination Rules}
\label{sec:gate-elim-rules}

\begin{figure}[H]
    \centering
    \begin{subfigure}[b]{0.23\textwidth}
       \centering
       \includegraphics[width=\textwidth]{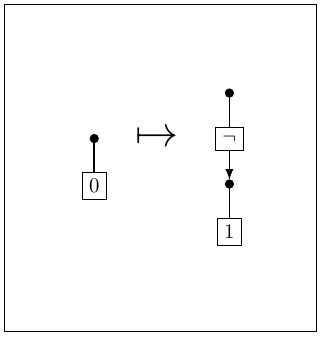}
       \caption{Zero Elim}
    \end{subfigure}
    \begin{subfigure}[b]{0.23\textwidth}
       \centering
       \includegraphics[width=\textwidth]{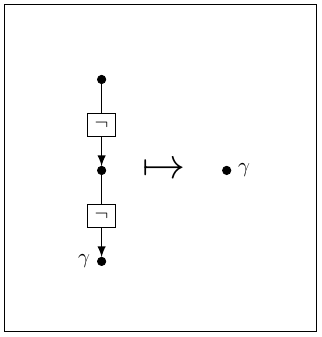}
       \caption{2xNegation Elim}
    \end{subfigure}
        \begin{subfigure}[b]{0.23\textwidth}
       \centering
       \includegraphics[width=\textwidth]{figures/pdf/rules/normalizing-and-dedup.pdf}
       \caption{AND-De-dup}
    \end{subfigure}
        \begin{subfigure}[b]{0.23\textwidth}
       \centering
       \includegraphics[width=\textwidth]{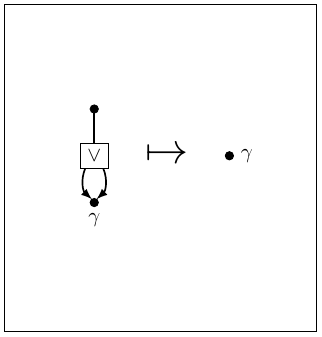}
       \caption{OR-De-dup}
    \end{subfigure}
    \caption{Normalizing Rules}
    \label{fig:normalizing-rules}
\end{figure}

\begin{figure}[H]
    \centering
    \begin{subfigure}[b]{0.23\textwidth}
       \centering
       \includegraphics[width=\textwidth]{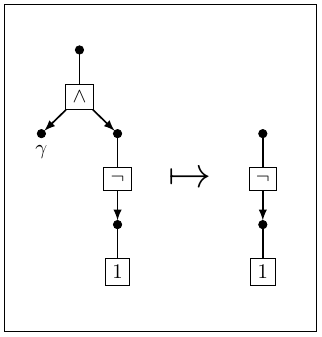}
       \caption{AND-Left}
    \end{subfigure}
    \begin{subfigure}[b]{0.23\textwidth}
       \centering
       \includegraphics[width=\textwidth]{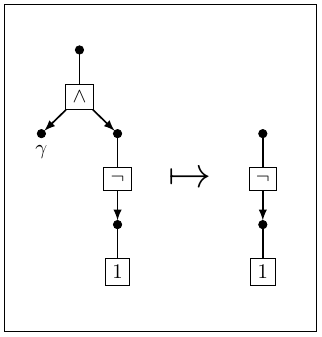}
       \caption{AND-Right}
    \end{subfigure}
        \begin{subfigure}[b]{0.23\textwidth}
       \centering
       \includegraphics[width=\textwidth]{figures/pdf/rules/fixing-or-left.pdf}
       \caption{OR-Left}
    \end{subfigure}
        \begin{subfigure}[b]{0.23\textwidth}
       \centering
       \includegraphics[width=\textwidth]{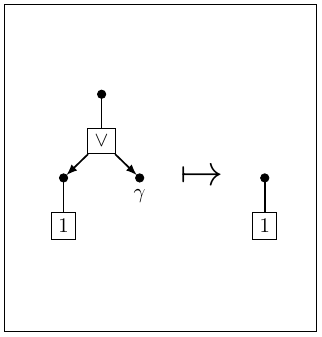}
       \caption{OR-Right}
    \end{subfigure}
    \caption{Fixing Rules}
    \label{fig:fixing-rules}
\end{figure}

\begin{figure}[H]
    \centering
    \begin{subfigure}[b]{0.23\textwidth}
       \centering
       \includegraphics[width=\textwidth]{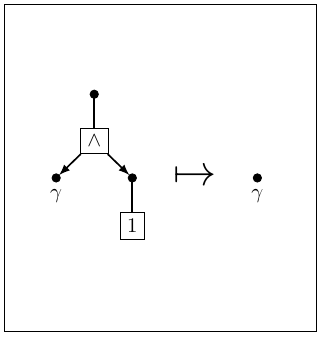}
       \caption{AND-Left}
    \end{subfigure}
    \begin{subfigure}[b]{0.23\textwidth}
       \centering
       \includegraphics[width=\textwidth]{figures/pdf/rules/passing-and-right.pdf}
       \caption{AND-Right}
    \end{subfigure}
        \begin{subfigure}[b]{0.23\textwidth}
       \centering
       \includegraphics[width=\textwidth]{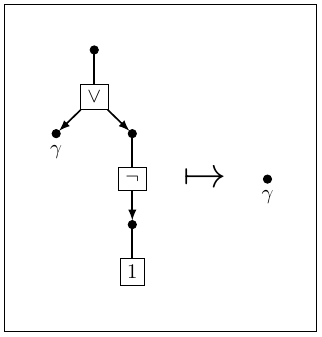}
       \caption{OR-Left}
    \end{subfigure}
        \begin{subfigure}[b]{0.23\textwidth}
       \centering
       \includegraphics[width=\textwidth]{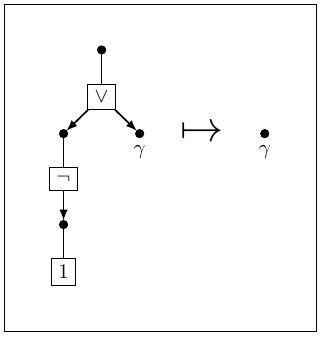}
       \caption{OR-Right}
    \end{subfigure}
    \caption{Passing Rules}
    \label{fig:passing-rules}
\end{figure}

\begin{figure}[H]
    \centering
    \begin{subfigure}[b]{0.23\textwidth}
       \centering
       \includegraphics[width=\textwidth]{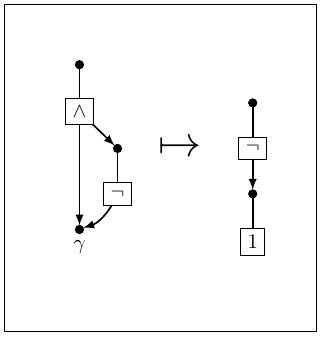}
       \caption{AND-Left}
    \end{subfigure}
    \begin{subfigure}[b]{0.23\textwidth}
       \centering
       \includegraphics[width=\textwidth]{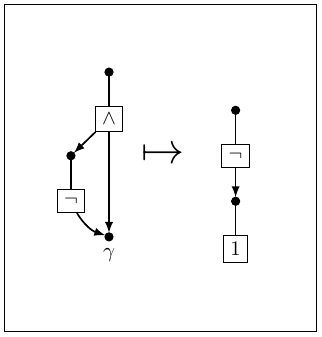}
       \caption{AND-Right}
    \end{subfigure}
        \begin{subfigure}[b]{0.23\textwidth}
       \centering
       \includegraphics[width=\textwidth]{figures/pdf/rules/tautology-or-left.pdf}
       \caption{OR-Left}
    \end{subfigure}
        \begin{subfigure}[b]{0.23\textwidth}
       \centering
       \includegraphics[width=\textwidth]{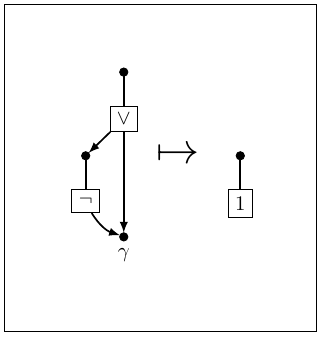}
       \caption{OR-Right}
    \end{subfigure}
    \caption{Tautology Rules}
    \label{fig:tautology-rules}
\end{figure}
\section{Circuit Size When NOT Gates are Free: Comparing Different Bases}
\label{sec:demorgan-u2-equivalence}

Recall that $\mathcal{U}_2$ is the basis consisting of all two-input Boolean functions except binary exclusive-or $\oplus$ and equivalence $\odot$ while the DeMorgan basis is $\{\land, \lor, \neg\}$. We show explicitly that these bases are equivalent in terms of size and depth when $\neg$ gates are not counted towards these complexity measures. The proof that $\cB_2$ and $\{\land, \oplus, \neg\}$ are equivalent when the latter's $\neg$ gates are free is identical.

\begin{lemma}[Folklore]\label{lem:demorgan-u2-equivalence}
    Let $f$ be any non-degenerate Boolean function besides $f(x) = \neg x$. Then the size and depth complexity of $f$ in the $\mathcal{U}_2$ ($\cB_2$) basis and the DeMorgan ($\{\land, \oplus, \neg\}$) basis where $\neg$ gates are free is equal.
\end{lemma}

\begin{proof}
    Let $CC_U(f)$ denote the circuit size complexity of $f$ in the $\mathcal{U}_2$ basis. Let $CC_D(f)$ denote the circuit size complexity of $f$ in the DeMorgan basis where $\neg$ gates are free.

    We begin by arguing $CC_D(f) \leq CC_U(f)$ by translating any optimal $\mathcal{U}_2$ circuit into a DeMorgan circuit of equal size. We recall the 14 binary Boolean operators that form the $\mathcal{U}_2$ in Table \ref{tab:u2-revisited}.

\begin{table}[H]
\centering
\begin{tabular}{|cc|cccccccccccccc|}
\hline
$p$ & $q$ & $\ast_1$ & $\ast_2$ & $\ast_3$ & $\ast_4$ & $\ast_5$ & $\ast_6$ & $\ast_7$ & $\ast_8$ & $\ast_9$ & $\ast_{10}$ & $\ast_{11}$ & $\ast_{12}$ & $\ast_{13}$ & $\ast_{14}$ \\ \hline
T & T & T     & F     & T     & F     & T     & F     & T     & F     & T     & F        & T        & F        & T        & F        \\
T & F & T     & F     & T     & F     & F     & T     & T     & F     & F     & T        & F        & T        & T        & F        \\
F & T & T     & F     & F     & T     & T     & F     & F     & T     & T     & F        & F        & T        & T        & F        \\
F & F & T     & F     & F     & T     & F     & T     & T     & F     & T     & F        & F        & T        & F        & T        \\ \hline
\end{tabular}
\caption{The $\cU_2$ basis consists of all the binary Boolean operations except $\XOR_2$ and $\neg \XOR_2$}
\label{tab:u2-revisited}
\end{table}

We claim that no optimal circuit for $f$ uses $\ast_i$ gates for $i \in [6]$.
For $\ast_1$ and $\ast_2$, replacing these gates with the constants $0$ and $1$ would reduce the size of the circuit. Similarly, $\ast_3$ and $\ast_5$ gates are non-optimal since they just compute the first and second inputs respectively. If a $\ast_3$ gate appears in a circuit, we can instead feed it's first input into it's outputs for the same effect. The gates $\ast_4$ and $\ast_6$ simply negate their first and second inputs respectively. Similarly, if a $\ast_4$ or $\ast_6$ appear internally we can pass the first and second inputs up to the gates' outputs but we must change the labels of these gates. We record the transformations in Table \ref{tab:U2} below:

\begin{table}[H]
\centering
\begin{tabular}{|cc|cccccccc|}
\hline
       &        & $\ast_7$  & $\ast_8$  & $\ast_9$  & $\ast_{10}$ & $\ast_{11}$ & $\ast_{12}$ & $\ast_{13}$ & $\ast_{14}$ \\ \hline
$\neg p$ & $q$      & $\ast_{12}$ & $\ast_{11}$ & $\ast_{13}$ & $\ast_{14}$   & $\ast_8$    & $\ast_7$    & $\ast_9$    & $\ast_{10}$   \\
$p$      & $\neg q$ & $\ast_{13}$ & $\ast_{14}$ & $\ast_{12}$ & $\ast_{11}$   & $\ast_{10}$   & $\ast_9$    & $\ast_7$    & $\ast_8$  \\ \hline
\end{tabular}
\captionof{table}{Gate substitutions when removing superfluous internal $\ast_4$ or $\ast_6$ gates \label{tab:U2}}
\end{table}

For example, if a $\ast_4$ gate is the first input of $\ast_7$, then we can remove the $\ast_4$ gate, feed its' first input into the first input of the $\ast_7$ gate and then relabel the $\ast_7$ gate to be a $\ast_{12}$ gate. If a $\ast_4$ or $\ast_6$ gate is the output of the circuit, then it must be negating one of the remaining $8$ Boolean functions since $f(x) \neq \neg x$. We can apply the transformations listed in Table \ref{tab:negations}.

\begin{table}[H]
\centering
\begin{tabular}{|c|c|c|c|c|c|c|c|c|}
\hline
& $\neg(\ast_7)$           & $\neg(\ast_8)$            & $\neg(\ast_9)$           & $\neg(\ast_{10})$         & $\neg(\ast_{11})$    & $\neg(\ast_{12})$             & $\neg(\ast_{13})$   & $\neg(\ast_{14})$              \\ \hline
Is Equivalent To & $\ast_8$ & $\ast_7$ & $\ast_{10}$ & $\ast_9$ & $\ast_{12}$ & $\ast_{11}$ & $\ast_{14}$ & $\ast_{13}$  \\ \hline
\end{tabular}
\captionof{table}{Gate substitutions for removing superfluous negations at the output \label{tab:negations}}
\end{table}

Lastly, the remaining $8$ binary Boolean functions in the basis can all be represented in the DeMorgan basis with only one binary gate:

\begin{table}[H]
\centering
\begin{tabular}{|c|c|c|c|c|c|c|c|c|}
\hline
& $\ast_7$           & $\ast_8$            & $\ast_9$           & $\ast_{10}$         & $\ast_{11}$    & $\ast_{12}$             & $\ast_{13}$   & $\ast_{14}$              \\ \hline
Is Equivalent To & $p \lor \neg q$ & $\neg p \land q$ & $\neg p \lor q$ & $p \land \neg q$ & $p \land q$ & $\neg p \lor \neg q$ & $p \lor q$ & $\neg p \land \neg q$ \\ \hline
\end{tabular}
\captionof{table}{Translations Between $\cU_2$ and DeMorgan formulas for $\ast_7$ through $\ast_{14}$ \label{tab:u2-to-demorgan}}
\end{table}

To show $CC_U(f) \leq CC_D(f)$ we show how to take any optimal DeMorgan circuit and transform it into an equivalent size $\mathcal{U}_2$ circuit. To do so, if any $\neg$ gate has fanout $m$ where $m > 1$, we split the $\neg$ gate up into $m$ copies, each with fanout exactly one. We also remove any double negations. Neither of these transformations change the size of the circuit. 

Then, in increasing topological order, we replace each $\land$ and $\lor$ gate with one of the $\mathcal{U}_2$ gates depending on whether their left or right inputs is negated. There are $8$ possible configurations, and the correspondence can be seen in Table \ref{tab:u2-to-demorgan}. Lastly, if the entire circuit is negated (i.e. the final output gate is $\neg$), then we absorb it into the $\mathcal{U}_2$ gate below using the transformations listed in Table \ref{tab:negations}.

Since $CC_U(f) \leq CC_D(f)$ and $CC_D(f) \leq CC_U(f)$ we get $CC_U(f) = CC_D(f)$. Notice that none of these transformations change the \emph{depth} of the circuits as well, and thus the bases are equivalent with respect to depth complexity as well.
\end{proof}
\section{Circuit Simplification in Action: Proving Schnorr's Lower Bound}
\label{sec:ckt-rewrite-demo}

In this section, we provide the full proof of Theorem \ref{thm:schnorr}, Schnorr's lower bound for $\XOR_n$, using our system $\cS$ to demonstrate that this system is powerful enough to capture traditional gate elimination arguments. As a reminder, the theorem states that any DeMorgan circuit $C$ computing $\XOR_n$ requires at least $3(n - 1)$ costly gates \cite{Schnorr74}.

As a quick reminder, proofs by gate elimination often use the following argument to show that a circuit $C$ has property $\cP$.

\begin{enumerate}
    \item Assume that $C$ does \textbf{not} have property $\cP$.
    \item Select a variable $x_i$ and constant $\alpha$ for substitution $\{ x_i \mapsto \alpha \}$ using $\neg \cP$ and the structure of $C$.
    \item \textbf{\emph{Simplify}} $C$ under the substitution $\{ x_i \mapsto \alpha \}$, to obtain a constant-free circuit $C'$.
    \item Argue that a critical property $\cP'$ of $C'$ implies a contradiction, therefore $C$ \textbf{must} have property $\cP$.
\end{enumerate}

System $\cS$ formalizes the simplification procedure used in step three above.  Usually, this is not necessary: the critical property is something like $\cP' = $ ``simplification eliminated four gates'' and it is clear that every sequence of simplification steps reaches a circuit $C'$ with property $\cP'$.  However, to algorithmically process the results of gate elimination we must assert more complicated post-simplification properties like $\cP' = $ ``input $x_j$ has fanout one,'' where $x_j$ was the sibling of $x_i$ in $C$.  These more delicate properties are not so easily seen to hold after \emph{every} terminated simplification.  Furthermore, we often wish to carefully sequence and analyze only the \emph{first few steps} of gate elimination, and then apply ``all remaining simplifications'' without considering them in detail.  It is critical that $\cP'$ emerge no matter how the steps of elimination are sequenced.

The \emph{convergent} simplification procedure $\cS$ for circuits has the following property: every valid run of $\cS$ on $C$ with $\alpha$ substituted for any input $x_i$ terminates with the \emph{same} circuit $C'$, up to bisimulation.  Therefore, to carry out the gate elimination argument template above, one need only exhibit a particular run of $\cS$ and argue that the resulting $C'$ has critical property $\cP'$.

In Appendix \ref{sec:schnorr-structured-proof}, we provide an alternative proof of Theorem \ref{thm:schnorr} using the structured proof format of Lamport \cite{lamport1995write,lamport2012write}. Due to the case analysis and repeated proofs by contradiction present, this alternative format may be easier to verify. Furthermore, we believe that this presentation style makes the intricate case analysis present both in our proof of Schnorr more explicit and readable. Furthermore, this style of proof is more amenable to verification using a computer.

\begin{proof}[Proof of Theorem \ref{thm:schnorr}]
  Let $C$ be an optimal circuit for $\XOR_n$ with $n\geq 2$. In the original proof, the goal was to find some \textbf{setting} of an input \textit{node} such that gate elimination would remove at least $3$ \textit{costly} ($\land$ and $\lor$) \textit{gate}. Besides notational differences, our goal remains the same. We will find  a \textbf{substitution} of an input \textit{hyperedge} which causes at least $3$ \textit{costly} gate \textit{hyperedges} (those labeled $\land$ or $\lor$) to be removed during \textbf{rewriting}. Following Schnorr we build up the circuit locally around an input by repeated proofs by contradiction; we perform substitutions and rewrites to contradict $C$'s optimality or the downward self-reducibility of $\XOR_n$ thereby forcing $C$ to have the desired structure.

  In order to smooth the transition from viewing circuits as graphs to viewing them as term graphs, we will simply refer to input hyperedges as inputs and gate hyperedges as gates. We describe the substitution and rewriting steps at a high level.

  We wish to first sort the gates in ``topological order". In the traditional view of circuits as DAGs (where gates are nodes) this notion is straightforward. However, since our gates are edges, we must do so indirectly. We can sort the vertices of $C$ topologically and, since each node is the result node of a unique edge, we then order the gates according to their result nodes. From this point on when we refer to sorting inputs or gates topologically, formally we are doing this process.
  
  Fix a topological order and let $h$ be the first costly gate in $C$. Under the traditional view of circuits, we would conclude $h$ has $x_i$ (or $\neg x_i)$ and $x_j$ (or $\neg x_j$) as inputs for some $i, j \in [n]$. Formally, this means $h$ has two argument nodes whose terms are $x_i$ (or $\neg x_i$) and $x_j$ (or $\neg x_j$). Before we continue, however, we streamline our verbiage. We notate the possibility of $\neg$ gates by defining the shorthand $(\neg) f$ to mean ``$f$ (or $\neg f$)." If $f'$ has an argument node whose term is $(\neg) f$ we say that $f$ \emph{feeds} into $f'$ and that $f'$ is a \emph{successor} of $f$. Lastly if the label of a gate $f$ is $\land$ or $\lor$ we say $f$ is \emph{costly}. Combining these allows us to instead say $h$ is the costly successor of two inputs $x_i$ and $x_j$ for some $i, j \in [n]$ --- maintaining the formalism of our system while being more inline with the original proof.

  We assert $i \neq j$. Otherwise $h \equiv (\neg) x_i \diamond (\neg) x_i$ for some $\diamond \in \{\land, \lor\}$. If this occurred, we could apply a normalizing or tautology rule; rewriting $C$ would then delete $h$. Since $h$ is a costly gate this would decrease the the size of $C$, contradicting $C$'s optimality.

  We now wish to say that $x_i$ has \textit{fanout} at least $2$ where we define \emph{fanout} to the number of \emph{costly} successors a gate or input has. Again, assume the contrary: $h$ is $x_i$'s only costly successor. We can then substitute $x_j = \alpha$ where $\alpha$ is set so that during rewriting, we can apply a fixing rule to $h$. This would mean that $C|_{x_j = \alpha}$ does not depend on $x_i$ violating Fact \ref{fact:xor-non-degeneracy} (All Subfunctions of $(\neg)\XOR$ are Non-Degenerate).

  Let $f$ be another costly successor of $x_i$. We can conclude that $(\neg) f$ is not the output of the circuit (i.e. the root node). If it were, then we could substitute $x_i = \beta$ and fix $(\neg)f$ during rewriting. This would fix the output of the circuit so that $C|{x_i = \beta}$ is constant contradicting Fact \ref{fact:xor-dsr} ($(\lnot) \XOR$ is fully DSR).
  
  Let $f'$ be a costly successor of $f$ and observe $h \neq f'$ since $f < f'$ in our topological ordering and $h$ was the first costly gate in the ordering. We now observe that if we set $x_i = \beta$ so that $f$ is fixed, then during rewriting we eliminate $h$ (using a fixing or passing rule), $f$ (using a fixing rule), and $f'$ (using a fixing or passing rule). This is because once $f$ is fixed, then the fixing $(\neg) 1$ now feeds into $f'$. In this case we say that $(\neg) 1$ \textit{inherited} $f'$ as a successor after this rewrite applies, and thus another rewriting rule will apply deleting $f'$.

  Thus $\XOR_{n}$ requires at least $3$ more costly gates than $\XOR_{n-1}$. Since $\XOR_{1}$ requires zero costly gates, we have using induction that $\XOR_{n}$ requires at least $3(n-1)$.
\end{proof}

\subsection{Structured Proof of Schnorr}
\label{sec:schnorr-structured-proof}

Below, we prove Theorem \ref{thm:schnorr} with a structured proof as proposed by Lamport \cite{lamport1995write,lamport2012write}. We believe that this presentation style makes the intricate case analysis present in this proof more explicit and readable. Furthermore, this style of proof is more amenable to verification using a computer.

\begin{proof}[Structured Proof of Theorem \ref{thm:schnorr}]

  Let $C$ be an optimal circuit for $\XOR_n$ where $n \geq 2$.
  \begin{enumerate}
  \item Let $h$ be the \emph{first} AND,OR gate of $C$ in topological order, so $h$ has $(\neg)x_i, (\neg)x_j$ as inputs for $i,j \in \nat$.

  \item $i \neq j$
    \begin{enumerate}
    \item Suppose not, so $i = j$
    \item Then $h \equiv (\neg)x_i \diamond (\neg)x_i$ where $\diamond \in \{\land , \lor\}$
    \item Thus, one of the normalizing or tautology rules from Gate Elim TGRS matches $h$
    \item \textbf{Rewrite} $C$ finding $h$ deleted
    \item Contradiction to optimality of $C$
    \end{enumerate}

  \item The fanout of $x_i$ must be at least 2.
    \begin{enumerate}
    \item Suppose not, so fanout of $x_i$ is 1.
    \item \textbf{Substitute} $x_j = \alpha$ in $C$ to fix $h$
    \item \textbf{Rewrite} $C$, finding that fanout of $x_i$ is now 0
    \item Thus, $C \restrict_{x_j = \alpha}$ does not depend on $x_i$
    \item Contradiction to Fact \ref{fact:xor-non-degeneracy}.
    \end{enumerate}

  \item Let $f$ be the other gate taking $x_i$ as input so $f \neq h$

  \item $(\neg)f$ is not the output gate of $C$.
    \begin{enumerate}
    \item Suppose it is, so $(\neg)f$ is the output gate of $C$.
    \item \textbf{Substitute} $x_i = \alpha$ in $C$ to fix $f$
    \item \textbf{Rewrite} $C$, finding that output of $C$ is constant
    \item Thus, $C \restrict_{x_i = \alpha}$ is a constant function
    \item Contradiction to Fact \ref{fact:xor-dsr}.
    \end{enumerate}

  \item Let $f'$ be a \emph{costly} successor of $f$ in $C$, such must exist.
    
  \item Eliminate three distinct gates with a substitution.
    \begin{enumerate}
    \item \textbf{Substitute} $x_i = \alpha$ in $C$ to fix $f$
    \item \textbf{Rewrite} $C$, finding at least $f, h, f'$ deleted
    \item \emph{argument:} Observe that $x_i$ ``touches'' gate $h$ to eliminate, and it fixes $f$ which ``touches'' gate $f'$
    \item there exists a 1-bit restriction eliminating $\ge 3$ gates
    \end{enumerate}

  \item Conclude $\XOR_n$ requires at least $3$ more costly gates than $\XOR_{n-1}$.
  \item Observe $\XOR_1$ requires $0$ costly gates.
  \item Use induction to show $\XOR_n$ requires at least $3(n-1)$.
\end{enumerate}  
\end{proof}

\end{document}